\documentclass[11pt]{article}

\usepackage{amsmath,amsfonts,amssymb,latexsym,amsthm,verbatim,a4wide,graphicx,subfigure}

\usepackage{
}  
\usepackage[english]{babel}
\usepackage[applemac]{inputenc}
\usepackage[square,numbers,sectionbib]{natbib}
\usepackage{chapterbib}

 \theoremstyle{plain}
\newtheorem{theorem}{Theorem}[section]
\newtheorem{lemma}[theorem]{Lemma}
\newtheorem{corollary}[theorem]{Corollary}

\theoremstyle{definition}
\newtheorem{definition}[theorem]{Definition}

\newtheorem{remark}[theorem]{Remark}
 
\newcommand{\pr}{{\mathbb P}}
\newcommand{\vc}[1]{{\mathbf #1}}

\newcommand{\suc}{{\rm success}}

\newcommand{\algdd}{{\tt DD\;}}
\newcommand{\algscomp}{{\tt SCOMP\;}}
\newcommand{\algcbp}{{\tt CBP\;}}
\newcommand{\algcomp}{{\tt COMP\;}}
\newcommand{\algsss}{{\tt SSS\;}}

\newcommand{\algCOMP}{\texttt{COMP}}
\newcommand{\algDD}{\texttt{DD}}
\newcommand{\algSCOMP}{\texttt{SCOMP}}
\newcommand{\algSSS}{\texttt{SSS}}
\newcommand{\beff}{\beta_{{\rm eff}}}
\newcommand{\SSS}{\text{\algSSS}}
\newcommand{\COMP}{\text{\algCOMP}}
\newcommand{\mDD}{\text{\algDD}}

\newcommand{\pd}{{\mathcal{PD}}}
\newcommand{\CC}{{\mathcal{C}}}
\newcommand{\DD}{{\mathcal{DD}}}
\newcommand{\subm}{{\mat S}}
\newcommand{\bino}{{\rm Bin}}
\newcommand{\binomi}[3]{b(#3; #1, #2)}
\newcommand{\ZZ}{{\mathbb Z}}

\newcommand{\E}{{\mathbb E}}

\newcommand{\df}[2]{\displaystyle{\frac{#1}{#2}}}
\newcommand{\veps}{\varepsilon}

\newcommand{\N}{{\mathcal{N}}}
\newcommand{\K}{\mathcal{K}}
\renewcommand{\L}{\mathcal{L}}
\newcommand{\Khat}{\widehat{\mathcal{K}}}
\newcommand{\Khsub}[1]{\widehat{\mathcal{K}}_{#1\!\!}}
\newcommand{\Nset}{\{1,\dots,N\}}

\newcommand{\zero}{\mathtt{0}}
\newcommand{\one}{\mathtt{1}}
\newcommand{\zo}{\{\mathtt{0},\mathtt{1}\}}

\newcommand{\compl}{\complement}
\newcommand{\DND}{\mathcal{ND}}

\newcommand{\defn}[1]{\emph{#1}}
\newcommand{\mat}[1]{\mathsf{#1}}
\renewcommand{\vec}[1]{\mathbf{#1}}
\newcommand{\ol}[1]{\overline{#1}}

\newcommand{\olc}[1]{\ol{R}^*_{#1}}

\newcommand{\etal}{{\it et al. }}
\begin{document}

\title{Group testing algorithms: bounds and simulations}
\author{Matthew Aldridge\thanks{Heilbronn Institute for Mathematical Research, School of Mathematics, University of Bristol, University Walk, Bristol, BS8 1TW, UK.  Email: {\tt M.Aldridge@bristol.ac.uk}} \and Leonardo Baldassini\thanks{School of Mathematics, University of Bristol, University Walk, Bristol,  BS8 1TW, UK. Email: {\tt Leonardo.Baldassini@bristol.ac.uk}} \and Oliver Johnson\thanks{School of Mathematics, University of Bristol, University Walk, Bristol,  BS8 1TW, UK. Email: {\tt O.Johnson@bristol.ac.uk} }}
\date{\today}
\maketitle

\begin{abstract}
We consider the problem of non-adaptive noiseless group testing of $N$ items of which
$K$ are defective.
We describe four detection algorithms: the \algCOMP\ algorithm of Chan \emph{et al.}; two new algorithms,
\algDD\ and \algSCOMP, which require stronger evidence to declare an item defective; and an essentially
optimal but computationally difficult algorithm called \algSSS.
By considering the asymptotic \emph{rate} of these algorithms with Bernoulli designs we see that \algDD\ outperforms
\algCOMP, that \algDD\ is essentially optimal in regimes where $K \geq \sqrt N$, and that no algorithm
with a nonadaptive Bernoulli design can perform as well as the best non-random adaptive designs when
$K > N^{0.35}$. In simulations, we see that \algDD\ and \algSCOMP\ far outperform \algCOMP, with
\algSCOMP\ very close to the optimal \algSSS, especially in cases with larger $K$.
%
%
\end{abstract}

\tableofcontents

\section{Introduction}

\subsection{General introduction}

Group testing is a combinatorial optimisation problem that was  introduced 
by Dorfman  \cite{dorfman43}  and has since given rise to the development of numerous  algorithms for its solution. This has included recent interest in
non-combinatorial, probabilistic methods  to tackle the problem.
Recently, the development of compressed sensing (see \cite{candes2006} for an introduction), has made group testing an object of renewed interest, since the two problems
can be viewed within a common framework of sparse inference (see \cite{malyutov,aksoylar}).
 The increasing awareness that other algorithmic problems may be reduced to group testing, and possibly be solved efficiently, encourages further study of the mathematical properties of this problem, increasing the understanding we have of it and creating analogies to other better understood problems.

 The group testing problem is traditionally exemplified by the application that first motivated it \cite{dorfman43}. Suppose a few soldiers within an army suffer from a certain infectious disease.  One could test every single individual, giving a time-consuming and possibly costly procedure. To reduce the cost, we divide the soldiers into a collection of subsets, pool the blood samples drawn from all soldiers in each subset and then test the pooled blood. Assuming the testing procedure is not subject to errors, obtaining a negative test implies that all soldiers in the relevant pool are healthy, whereas a positive test indicates that at least one soldier in the pool is infected. 
We wish to  minimise the number of tests required subject to the success probability of our procedure being high.

More formally, we consider a set $\N = \Nset$ of $N$ \defn{items}, of which a subset $\K
\subset \N$ are \defn{defective}. We will write $K = |\K|$ for
the number of defectives. Note, however, that none of our detection algorithms 
require knowledge of $K$, or even bounds on $K$, in order to estimate the defective set. However the derivation of 
our bounds on rate and success probability of these algorithms will depend on $K$. We will
assume throughout, though, that defectivity is rare, in that $K \ll
N$. Of course, if defectivity is not rare, a strategy of testing each
item individually will be both effective and extremely simple.

To perform nonadaptive group testing, an experimenter needs to decide
on two things. First, in what we shall call the \defn{design stage}
they must design testing pools, by deciding which items will be
included in which tests. Second, in what we shall call the
\defn{detection stage}, they must use the results of the pooled tests
to detect which items were defective. Nonadaptive algorithms differ from 
adaptive algorithms in that the latter alternate design and detection steps,
exploiting the information gathered after each test to design future ones. In
nonadaptive group testing, on the other hand, all the tests are designed
a priori and then carried out concurrently.

Much  work on the design stage of nonadaptive group testing has concentrated on
carefully constructing test designs with certain properties (known as
\defn{disjunctness} and \defn{separability}, see Definition \ref{def:disjunct}) that will with certainty
detect the defective set in $T$ tests as long as the number of
defective items is no more than $K$, for some predecided $K$ and $T$.
Given such a test design, the detection stage is usually simple
\cite[Chapter 7]{du} \cite{dyachkov2}. 

However, such designs can be unsuitable for practical situations. For
example, it assumes that the experimenter either knows $K$ or has an
upper bound on the number of defectives before the experiment begins.
Also, if the experimenter is unable to carry out all $T$ tests, there
will be no guarantees on the performance of the procedure; and
conversely, if the experimenter is able to perform some extra tests,
the procedure is unlikely to be able to take advantage of them. 
Further (see for example \cite[Chapter 7]{du}, \cite{dyachkov2})
these designs give performance that does not meet information theoretic
bounds such as Theorem \ref{thm:basic} below.

This has led to interest in simpler designs, such as the Bernoulli($p$)
random design, where each item is in each test independently at random
with some probability $p$. Work that uses these designs includes
\cite{chan2011}, \cite{malioutov2012}, \cite{atia},
\cite{johnsonc8}, \cite{aldridge}.
 This random design
does not  require the experimenter to understand and accurately implement
tricky combinatorial designs, as it does not necessarily require accurate
knowledge of the number of defectives, or how many tests will be
performed. Furthermore, recent work by Atia and Saligrama
\cite{atia} has shown that the Bernoulli$(1/K)$ design is
asymptotically close to optimal when $K = o(N)$.

\subsection{Paper outline}

In this paper we study four detection algorithms for group testing, which
we explain fully in Section \ref{sec:algorithms}:
\begin{description}
  \item[Combinatorial optimal matching pursuit] (\algCOMP), a simple algorithm due to Chan et al \cite{chan2011,chan2}.
  \item[Definite defectives algorithm] (\algDD), a new algorithm, which is similar
    to \algCOMP, but requires stronger evidence to declare an item as defective.
  \item[Sequential \algCOMP] (\algSCOMP), a new algorithm that starts with \algDD, but
    marks extra items as defective in a sequential manner, ensuring the result is
    a \emph{satisfying set} (see Defintion \ref{def:ss}).
  \item[Smallest satisfying set] (\algSSS), a `best possible' algorithm, albeit
    one that is unlikely to be computationally feasible for large problems (although we
    do discuss how using \algDD\ as a preprocessing step may make it plausible in
    regimes where \algDD\ performs reasonably well).
\end{description}

Although we believe these algorithms should work well for a variety of test designs, we are particularly intereseted in their performance with the popular Bernoulli random design.

In Section \ref{sec:rates}, we analyse the algorithms by deriving bounds on their \emph{maximal achievable rate} (Definition \ref{def:cap})
in different sparsity regimes. First, we see that our new \algdd algorithm achieves higher rates than the \algCOMP\ 
algorithm in all sparsity regimes (except in most sparse regime where $K$ is fixed,
when they perform equally). Second, we see that \algDD\ performs as well as \algSSS\ in the more dense
regimes where $K\geq \sqrt{N}$, and hence that its performance is asymptotically essentially optimal in those cases.
Third, we note that, in denser cases where $K > N^{0.35}$, even the \algSSS\ algorithm falls short of what is achievable
with nonrandom adaptive testing -- suggesting either that Bernoulli test designs are suboptimal for nonadaptive testing in that regime, or
that there exists an `adaptivity gap' between what is possible by adaptive and nonadaptive testing. We summarise these results graphically in
Figure \ref{fig:capsketch}.

In Section \ref{sec:simulation} we perform simulations on the algorithms. We see (in Fig.~\ref{graph1}, for example) that \algDD\ far outperforms the
\algCOMP\ algorithm, and the \algSCOMP\ performs better still -- very close to the impractical but optimal
\algSSS\ algorithm.

\subsection{Previous work}
We now give an overview of some previous work on noiseless non-adaptive group testing. As mentioned above, we have been observing an increasing curiosity about the structural (and not simply algorithmic) properties of group testing. In fact, this dates back to the work of Malyutov and co-authors in the 1970s (see \cite{malyutov} for  a
review of their contribution), who
 established an analogy between noisy group testing and Shannon's channel coding theorem
\cite{shannon2}. The idea is to treat the recovery of the defective set as a decoding procedure for a message transmitted over a noisy channel, where the testing matrix represents the codebook used to translate the message. 
Using such ideas, more recent work of
Atia and Saligrama \cite{atia}  mimics the channel coding theorem's results and obtains an upper bound of $O(K\log N)$ on the number of tests required. 
Such an upper bound refers to the amount of tests needed for arbitrarily small average error probability, and should in fact be loosened depending on the kind of error produced by noise, e.g. false positives or negatives. Still following the information-theoretic path, Atia and Saligrama
\cite{atia} also prove a lower bound on the number of tests using Fano's inequality; unlike in the case of channel coding, the upper and lower bound seem not to meet asymptotically. Moreover, the authors also show that the same upper bound $T =  O(K\log N)$ holds even for noiseless group testing. Similar results 
had already been derived in the past, 
see for example the work of Malyutov \cite{malyutov3, malyutov4} in a very general setting.

Wadayama \cite{wadayama} describes an approach to the design phase of the group testing
problem motivated by LDPC codes. In particular, he chooses test matrices with 
constant row and column weights, and proves theoretical results which (in the regime where
$K = O(N)$) bound the optimal code size from above and below. In many cases (see \cite[Figure I]{wadayama} for details)
the resulting lower and upper bounds are very close; often within 5\% of each other, or even less.
 Notice that Wadayama's results should be compared with the densest
problems we consider ($\beta \sim 0$), where we show (see Figure \ref{fig:capsketch} for a summary)  that the \algdd algorithm performs close to its theoretical
optimum.
However, in \cite{wadayama} Wadayama does not discuss the question
of how decoding can be practically achieved.

In terms of decoding algorithms, the similarity between compressive sensing and group testing
(as discussed in \cite{malyutov,aksoylar}) has been used in \cite{chan2011,chan2} by Chan \etal to present testing algorithms for both noiseless and noisy non-adaptive group testing. In particular, the authors introduce the Combinatorial Basis Pursuit ({\algcbp}\!\!) and Combinatorial Orthogonal Matching Pursuit ({\algcomp}\!\!) algorithms, and their
 noisy versions ({\tt NCBP} and {\tt NCOMP}), prove universal lower bounds for the number of tests needed  to get a certain success probability and upper bounds for the algorithms they are introducing. The \algcomp 
algorithm allows  the strongest bounds in their paper to be rigorously proved, and will be the basis of our work.

Other approaches to classical instances of group testing have been proposed in the literature. In particular, its natural integer-programming (IP) formulation has been addressed by Malioutov and Malyutov \cite{malioutov2012}, Malyutov and Sadaka \cite{malyutov2} and Chan \etal \cite{chan2}: noticing that group testing allows an immediate IP formulation, it is possible to relax the integer program and solve the associated linear version
(see Section \ref{sec:sss}).
These authors then consider decoding algorithms that find integer solutions `near' (in some sense) to the relaxed solution.

\section{Definitions and notations}

We now formally define the main concepts and terminology we shall use in this paper

\begin{definition} \label{def:designs}
A \defn{test design} of $T$ \defn{tests} can be summarised by a
\defn{testing matrix} $\mat X = (x_{it} : i \in \N, t = 1,\dots, T)$,
where $x_{it} = \one$ indicates that item $i$ is included in test $t$
and $x_{it} = \zero$ indicates that item $i$ is not included in test
$t$.

A \defn{Bernoulli$(p)$ test design} is defined by the random testing
matrix $\mat X$ whose $(i, t)$th element $X_{it}$ is $\one$ with
probability $p$ and $\zero$ with probability $1-p$, independent over
$i$ and $t$.\end{definition}

As previously mentioned, past work on group testing focussed on constructing
test designs with the favourable structural properties of {\em disjunctness}
and {\em separability}. These properties are in practice very restrictive, and are defined as follows.

\begin{definition}\label{def:disjunct} Consider a testing matrix $\mat X\in\{\zero,\one\}^{T\times N}$, and recall we write $\N$ for the set of all items:
\begin{enumerate} \item
 $\mat X$
	 is called {\em $K$-disjunct} if, 
	for all subsets $\L\subset \N$ of cardinality $|\L|\leq K$:
	\begin{equation}  \text{for all  $i\in \L^\complement$ there is a test $t$ such that (a) 
	$x_{it}=\one$, and (b) $\mat x_{jt}=0$ for all $j\in \L$.} \label{eq:disjunct} \end{equation}
In particular, taking $\L = \K$, the true defective set, we see that $K$-disjunctness implies that
every non-defective item appears 
in at least one negative test. 
\item
	$\mat X$ is said to be {\em $K$-separable} if, denoting by $\vec x_i$ 
	the $i$-th column of $\mat X$, 
	for all pairs of distinct subsets $\mathcal I, \mathcal J \subset \N$ of cardinality 
	$|\mathcal I|, |\mathcal J|\leq K$, we have
	\[ \bigvee_{i\in \mathcal I} \vec x_i \neq \bigvee_{i\in \mathcal J}\vec x_i\ , \]
	where $\bigvee$ denotes the componentwise boolean sum of binary vectors (an \texttt{OR} operation).
\end{enumerate}
\end{definition}

The detection stage of an algorithm will be based on the outcomes of the tests.
The \defn{outcome} of a test will be positive if there is at least one
defective item in the test, and negative if there are no defectives in
the test. Formally:

\begin{definition}
If we write $y_t = \one$ for the outcome of the
$t$th test being positive and $y_t = \zero$ for it being negative, we
have
  \begin{equation}
 y_t = \begin{cases}
             \one & \text{ if\;\; $|\{ i \in \K : x_{it} = \one \}| \geq 1$,} \\
             \zero & \text{ if\;\; $|\{ i \in \K : x_{it} = \one \}| = 0$.}
           \end{cases} \end{equation}
It will be convenient to write $\vec y = (y_t) \in \{\zero, \one\}^T$
for the vector of all the outcomes.
\end{definition}

In other words, using the notation above, we have $\vec y = \bigvee_{i \in \K} \vec x_i$.

\begin{definition}
A \defn{detection algorithm} is a method to estimate the defective set
from the test outcomes; that is, a function $\Khat \colon \{\zero,
\one\}^T \mapsto {\mathcal P}(\N)$ (where we write ${\mathcal P}(\N)$ for the power set of $\N$), that associates to each outcome
vector $\vec y$ a subset $\Khat \subset \N$ of the items.
\end{definition}

It will be useful to write
  \[ \binom{\mathcal A}{n} := \{ \mathcal B \subset \mathcal A : |\mathcal B| = n \} \subset {\mathcal P}(\mathcal A) \]
to denote the subsets of a set $\mathcal A$ of size $n$.

\begin{definition}
The \defn{average error probability} is defined by
\begin{equation}
\epsilon := \pr_{\mat X, \K} (\Khat \neq \K)
               = \frac{1}{\binom NK} \sum_{\K \in \binom{\N}{K}} \pr_{\mat X} (\Khat \neq \K).
\end{equation}
Here, the probability is over the random defective
set $\K$ and, if a random test design is used, the random choice of
$\mat X$.  If $\mat X$ is
deterministic, then the summand is just an indicator function.

We write $\pr(\suc)  = 1 -\epsilon$ for the success probability.
\end{definition}


An important notion will be that of a \defn{satisfying set}.

\begin{definition}\label{def:ss}
Given a test design $\mat X$ and outcomes $\vec y$, we shall call a
set of items $\L \subset \N$ a \defn{satisfying set} if group testing with
defective set $\L$ and test design $\mat X$ would lead to the outcomes
$\vec y$.

Clearly the defective set $\K$ itself is a satisfying set.
\end{definition}


The effectiveness of group testing algorithms often depends on the \emph{sparsity}
of the problem; that is, how common it is for items to be defective.
In this paper, for benchmarking purposes, we consider a range of sparsity regimes, parameterised by a
\emph{sparsity parameter} $\beta$. Specifically, we consider
$K = N^{1 - \beta}$ for $0 < \beta \leq 1$. So large $\beta$ corresponds
to the most sparse cases, while small $\beta$ corresponds to the less sparse
(or denser) cases.
This sparsity parametrization was considered in different contexts by
Donoho and Jin \cite{donoho} and by Haupt, Castro and Nowak \cite{haupt}.


We will summarize the performance of our detection algorithms by considering their
\emph{maximum achievable rate} with Bernoulli tests and the full range of sparsity
regimes $\beta \in (0,1]$. Here, following \cite{johnsonc10}, the rate can be thought
of as the number of bits per test learned by the group testing algorithm.

\begin{definition} \label{def:cap}
Consider group testing with $N$ items of which $K$ are defective. An algorithm that
uses $T$ tests is said to have \emph{rate} $\log_2 \binom NK / T$.

A rate $R$ is said to be \emph{achievable} by an algorithm \texttt{A} in sparsity regime
$\beta$ if, for any $\delta > 0$, there is some group testing procedure with $N$~items,
$K = N^{1-\beta}$ defective items, when algorithm~$\texttt{A}$ uses $T$~tests, where the rate satisfies $\log_2 \binom{N}{N^{1-\beta}} / T \geq R$,
and the error probability satisfies $\epsilon \leq \delta$.

We write $R^*_{\mathtt{A}}(\beta)$ for the maximum achievable rate for algorithm \texttt{A}
in sparsity regime $\beta$, and define the \emph{capacity} $C(\beta)$ to be the maximum rate achievable by any group
testing algorithm in sparsity regime $\beta$.
\end{definition}

We note that a similar
concept of rate, defined for fixed $K$ as
$R(K) = (\log_2 N)/T$ was studied by Malyutov and others
\cite{malyutov}. This corresponds only to our sparsest regime $\beta = 1$, while our definition
allows us to make comparisons across a much wider sparsity range.

In this paper, for consistency, we will compare bounds on the success probability $\pr(\suc)$ and rate $R$ for different algorithms. In both cases, large values represent a more
successful algorithm. For example, we will refer to a result as a lower bound if it controls the rate and success probability from below (gives performance guarantees).



A simple counting argument (see for example Theorem \ref{thm:basic} of this paper) shows
that $C(\beta) \leq 1$.  For adaptive testing, in \cite{johnsonc10} it was shown that for $\beta > 0$,
we can indeed achieve the capacity $C(\beta)=1$ using the generalized binary splitting algorithm of Hwang \cite[Section 2.2]{du}. Analogous results 
for slightly different or more general settings are also present in the literature; see for example
\cite{malyutov5} for its particular focus on adaptive algorithms and 
references therein. 

In comparison, we shall see later that the essentially optimal \algSSS\ algorithm falls short of this in some denser regimes, in that
we certainly have $R^*_{\SSS}(\beta) < 1$ for $\beta < 0.65$ (see Theorem  \ref{sssthm} below). This could be because Bernoulli test designs
are suboptimal in these regimes, or it could be that no nonadaptive procedure can achieve rate $1$, meaning there
is an `adaptivity gap' for denser problems.

\section{Algorithms} \label{sec:algorithms}

In this section we explain the algorithms for the detection stage we will analyse in this paper.
The algorithms are intended to work for any test design, though we will usually analyse 
their performance in the context of Bernoulli test designs (see Definition \ref{def:designs}).

\subsection{Definite non-defectives -- \algcomp algorithm} \label{sec:comp}

A simple inference from noiseless group testing is the following:
if an item appears in a negative test, then it cannot be defective. This motivates the following definition:
\begin{definition} \label{def:nd}
We consider the guaranteed \defn{not defective} (ND)
set
\begin{equation} \label{eq:dnd} 
 \DND := \{i : \text{for some $t$ (a) $x_{it} = \one$ and (b) $y_t = \zero$} \}, \end{equation}
 and write $\pd = \DND^\compl$ for the set of \defn{possible defectives} (PD). \end{definition}
Chan \etal \cite{chan2011} suggest an algorithm, which they call combinatorial
orthogonal matching pursuit (\algcomp\!\!), that takes the ND items to be
non-defective but all other items to be defective.\footnote{In
a later paper, Chan \etal refer to the algorithm as
\texttt{CoMa} (column matching) \cite{chan2}. The decoding part of
their \texttt{CBP} (combinatorial basis pursuit) \cite{chan2011} or
\texttt{CoCo} (coupon collector) \cite{chan2} algorithm works the same way,
although is only considered as applied to a slightly different random test design.}
That is, \algcomp takes as an estimate all possibly defectives, or $\Khsub{\algcomp} = \pd$.

Note that $\Khsub{\algcomp}$ is a satisfying set (in the sense of Definition \ref{def:ss}) -- in fact, it is the largest
satisfying set. Thus if the true defective set $\K$ is the unique satisfying
set then the \algcomp algorithm certainly finds it.  Note also that the \algcomp
algorithm can only make false-positive errors (declaring nondefective items to be defective), and never makes false-negative
errors (declaring defective items to be nondefective); in other words, we have $\Khsub{\algcomp} \supseteq \K$.

 Notice, moreover,
that by Definition \ref{def:disjunct}, if the design $\mat X$ is $K$-disjunct 
the \algcomp can successfully recover the defective set.
This is because $K$-disjunctness implies that every non-defective item appears 
in at least one negative test, hence there are no intruding non-defectives. However, notice that $K$-disjunctness is a very restrictive property, since it imposes restrictions on
all sets $S$ of cardinality $\leq K$, whereas \algcomp will succeed if   property (\ref{eq:disjunct}) holds for $S$ being the true defective set $\K$.

\subsection{Definite defectives -- \algdd algorithm} \label{sec:dd}

Once the possible defective (PD) items have been identified, some other elements can be identified
as being \defn{definitely defective} (DD).  
The key idea is that
if a positive test contains exactly one possible defective item,
 then we can in fact be certain  that item is  defective.
This motivates
our \algdd algorithm, which uses the possible defectives $\pd$ found in the \algcomp
algorithm 
as a starting point. 
The \algdd algorithm has three
steps:
\begin{enumerate}
	\item Define the possible defectives $\pd = \DND^{\compl}$, for
the set $\DND$ introduced in  (\ref{eq:dnd}). 
	\item 
For each positive test which contains a single item from $\pd$, declare the corresponding
item to be defective.
	\item All remaining items are declared to be non-defective.
\end{enumerate}
More formally, the \algdd algorithm defines every item in the set
  \begin{equation} \DD := \{i \in \pd: \text{for some $t$, (a) $x_{it} = \one$, (b) $x_{jt} = \zero$
for all $j \in \pd\setminus\{i\}$ and (c) $y_t = \one$}\} \label{eq:defdd} \end{equation}
to be defective, and all other items to be non-defective.
That is, we take $\Khsub{\algdd} = \DD$. Note that 
 $\Khsub{\algdd}$ need not be a satisfying set.

\begin{figure}
\begin{center}
\includegraphics{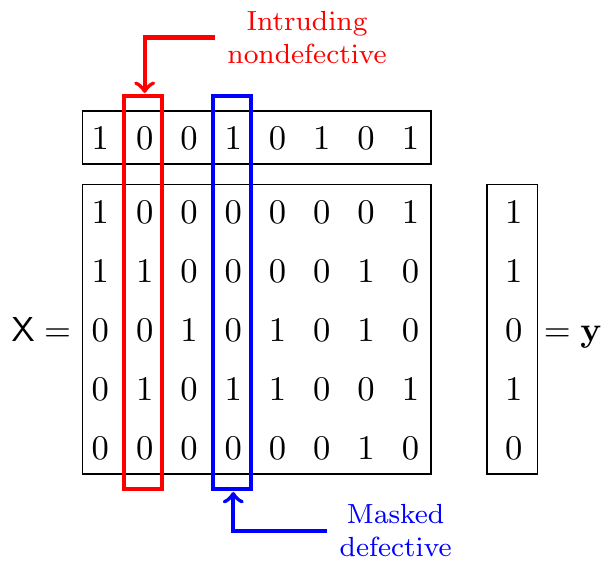}
\caption{An example of a group testing problem, including a masked defective and
an intruding non-defective, in the terminology we introduce here. The masked 
defective never appears in a positive test without some other defective item
also being present. The 
intruding non-defective never appears  in a negative test.}
\label{fig:masked}
\end{center}
\end{figure}

Notice that steps 1 and 2 in the \algdd algorithm make no mistake; indeed, step 1 just isolates all items that are ND, which can then be ignored, thus allowing us to restrict our attention to the items in $\pd$.
The set $\pd$ contains the $K$ true defectives, plus a (random) number $G$ of 
\defn{intruding} non-defectives (see Figure \ref{fig:masked}), meaning we can analyse the $T \times (K+G)$ 
submatrix $\subm$,  corresponding to the items in $\pd$.  
 Step 2, in turn, isolates the definitely defective items of $\pd$, i.e. those defectives that appear with no other item of $\pd$. After step 2 we are then left with
\begin {itemize}
	\item $G$ intruding non-defectives that haven't been discarded in step 1;
	\item defectives that never appear without other $\pd$ items in a test (we call
such an item {\em masked} -- see Figure \ref{fig:masked}).
\end{itemize}
Hence only step 3 can make a mistake, which occurs
when there are masked defectives which are erroneously
declared to be non-defective.
In other words,
the \algdd algorithm can only make false-negative errors, and never makes
false-positive errors, so $\Khsub{\algdd} \subseteq \K$.

The motivation for Step 3 of the DD algorithm comes from the sparsity of the defective set. That is, we cannot be sure whether items in $\pd$ but not $\DD$ are defective or not. However, since defectiveness is assumed to be  rare, in that $K \ll N$, it seems natural to assume that these items are nondefective, in the absence of evidence to the contrary. Conversely, the \texttt{COMP} algorithm assumes that these unknown items are defective, thereby often making false positive errors.

We derive an exact expression for the error probability of the  \algdd algorithm
in Section \ref{sec:ddsuc}.
The main difference between the \algdd algorithm
and the \algcomp algorithm
of Chan \etal \cite{chan2011} is that \algcomp succeeds if and only if $G=0$, whereas \algdd can succeed for positive $G$.

\subsection{The \algscomp algorithm} \label{sec:scomp}

In order to improve on the \algdd algorithm of Section \ref{sec:dd} we introduce
the \algscomp (Sequential \algCOMP) algorithm. The key
observation is that $\Khsub{\algdd}$ need not be a satisfying set, since there may
exist positive tests which contain no elements of $\Khsub{\algdd}$.
\begin{definition} Given an estimate $\Khat$, we say that a positive test is
{\em unexplained} by $\Khat$ if it contains no element from $\Khat$.

Note that a set
$\Khat \subseteq \pd$ of possible defectives being a satisfying set is equivalent to there being no unexplained positive
tests.
\end{definition}
Since each unexplained test must contain at least one of the masked defectives
in $\K \setminus \Khsub{\algdd}$, we might consider items in $\pd$ that appear in
many unexplained tests as most likely to be defective. The \algscomp algorithm
uses this principle to sequentially and greedily extend $\Khsub{\algdd}$ to a satisfying
set, by seeking items which explain the most currently unexplained tests.
 This is an attempt to exploit all the information available 
at each step, which is updated every time an item in $\pd$ is added to $\Khat$. 

The algorithm proceeds as follows:
\begin{enumerate}
\item Carry out
the first two steps of the \algdd algorithm; that is, generate an initial estimate
$\Khat = \Khsub{\algdd} = \DD$, for $\DD$ as defined in  (\ref{eq:defdd}).
\item \label{item:loop} Given an estimate $\Khat$:
\begin{enumerate}
\item If $\Khat$ is satisfying, terminate the algorithm, and use $\Khat$ as our 
final estimate of $\K$.
\item If $\Khat$ is not satisfying, then find the element $i \in \pd$ which appears
in the largest number of tests which are unexplained by $\Khat$ 
(breaking ties arbitrarily), and create a new estimate 
$\Khsub{\rm new} = \Khat \bigcup \{ i \}$. Repeat Step \ref{item:loop}.
\end{enumerate}
\end{enumerate}
Notice that Step \ref{item:loop} gives an iterative procedure which greedily
extends any estimate $\Khat$ to a satisfying set.
The \algscomp algorithm is hard to analyse theoretically. However in Section \ref{sec:simulation}
we give evidence from simulations that it outperforms the \algdd algorithm which it is based
on, and gives performance very close to optimal.

It is interesting to notice the analogy with Chvatal's approximation algorithm to the {\em set covering problem} (or just `set cover') --
see \cite{vazirani2004} for a discussion. Given a set $U$ and a family of subsets $S\subseteq
{\mathcal P}(U)$,
set cover requires us to find the smallest family of subsets 
 in $S$ whose union contains all elements of $U$. This optimisation problem is NP-hard, as for a putative solution optimality cannot be verified in polynomial time. In 1979, Chvatal \cite{chvatal1979} proposed an approximate solution by choosing, at each stage, the set in $S$ that covers the most uncovered elements. The algorithm produces a solution which can  be at most $H(|U|)$ times larger than the optimal, where $H(n) \sim \ln n$ 
is the $n$-th harmonic number. 
Similarly, \algscomp chooses defective items in a greedy manner to `cover' (or in our terminology, `explain') as many tests as possible, until all tests are explained. So similarly, we are guaranteed to find a satisfying set with no more than $KH(K) \sim K \ln K$ items.

Note that this implies that if the test design $\mat X$ is $KH(K)$-separable (see Definition \ref{def:disjunct}), then the there will be only one satisfying set constaining $KH(K)$ or fewer items. Since \algSCOMP\ finds such a satisfying set, in this situation it is guaranteed to find the correct defective set $\K$. As before, though, we note that \algSCOMP\ can succeed even when $KH(K)$-separability is not achieved.

There are inapproximability results that accompany Chvatal's algorithm, showing that, under standard complexity theory assumptions, no better approximation ratio is possible for set cover
(see for example
\cite[Theorem 29.31]{vazirani2004}).
In the light of this, we might consider \algSCOMP\ to similarly be a `best possible practical approximation' to the
\algSSS\ algorithm below.


\subsection{Smallest satisfying set -- \algsss algorithm} \label{sec:sss}


We now consider what an optimal detection algorithm might look like, without worrying about its computational feasibility.

Facts we know for sure about the true defective set $\K$ are:
\begin{itemize}
 \item $\K$ is a satisfying set, since we are considering noiseless testing,
 \item $K = |\K|$ is likely to be small, since we are considering regimes where $K \ll N$.
\end{itemize}

This suggests an approach where we attempt to find the smallest
set that satisfies the outputs. (This approach is similar to that taken in compressed sensing,
where one typically seeks the sparsest signal $\vec x$ that fits some given measurements $\vec y = \mat A \vec x$.)

That is, if we let $\vec z$ be
a solution to the $\zero$--$\one$ linear program
  \begin{align}
  \begin{split}
    \text{minimize} \qquad & \vec 1^\top \vec z  \\
    \text{subject to} \qquad & \vec x_t \cdot \vec z = 0    \quad \text{for $t$ with $y_t = \zero$}  \\
                             & \vec x_t \cdot \vec z \geq 1 \quad \text{for $t$ with $y_t = \one $}  \\
     & \vec z \in \zo^N ,
  \end{split} \label{ilp}
  \end{align}
then the smallest satisfying set \algsss algorithm uses
$\Khsub{\algsss} = \{i : z_i = \one \}$. (If there is not a
unique solution to \eqref{ilp}, choose one of the solutions arbitrarily.)
We analyse the success probability of the \algsss algorithm in Section \ref{sec:sssprob}.

Note that if the number of defectives $K$ is known, we can add the constraint
$\vec 1^\top \vec z \geq K$ to ensure we find a satisfying set of size exactly $K$. In this situation, the \algSSS\ algorithm finds an arbitrary satisfying set of the correct size, and since we are considering noiseless testing, one can do no better than this, so \algSSS\ is optimal. Hence, for the unknown-$K$ setting we consider in this paper, we will often refer to \algSSS\ as `essentially optimal'.

Furthermore, notice that if the test design $\mat X$ is $K$-separable (see Definition \ref{def:disjunct}), then the defective set is also the 
smallest satisfying set. Indeed, $K$-separability implies that no two sets of columns of $\mat X$ of size at most $K = |\K|$ have the same boolean sum, meaning that no other set as sparse as $\K$ or sparser than $\K$ could lead to the same outcome $\vec y$.

Unfortunately $\zero$--$\one$ linear programming is NP-hard, so the
\algsss algorithm is unlikely to be feasible for large problems.  We
include it here as a `best possible' benchmark against which to
compare other more feasible algorithms.

However, for moderately sized problems, we can use our new
\algDD\ algorithm as a preprocessing  step to reduce the size of the program
\eqref{ilp}.  Specifically, we can set
  \begin{align*}
    \mathcal N^* &:= \mathcal N \setminus (\DND \cup \DD) , \\
    \mathcal T^* &:= \{ t \in \{ 1, \ldots, T \} : x_{it} = \zero \text{ for
all $i \in \DD$, and } y_t = \one \} , \\
        \mat X^* &:= (x_{it} : i \in \mathcal N^*, t \in \mathcal T^*) ;
  \end{align*}
find $\vec z^* = (z_i^* : i \in \mathcal N^*)$ to solve the
smaller problem
  \begin{align}
    \text{minimize} \qquad & \vec 1^\top \vec z^* \notag \\
    \text{subject to} \qquad & \mat X^*\vec z^* \geq \vec 1 \notag \\     & \vec z^* \in \zo^{|\mathcal N^*|} ; \notag
  \end{align}
and choose
  \[ \Khsub{\algsss} = \DD \cup \{i \in \mathcal N^* : z_i^* = \one \} . \]

If the number of `not definitely anything' items $|\mathcal N^*|$ is only
of order $\ln N$, then the complexity of this problem becomes only polynomial in $N$,
and could be regarded as practical.

We also mention that recent work by Malyutov and coauthors
\cite{malioutov2012,malyutov2} has tried to construct the
defective set from
the solution to the relaxed problem
  \begin{align}
    \text{minimize} \qquad & \vec 1^\top \vec z \notag \\
    \text{subject to} \qquad & \vec x_t \cdot \vec z = 0    \quad \text{for $t$ with $y_t = \zero$}  \\
                             & \vec x_t \cdot \vec z \geq 1 \quad \text{for $t$ with $y_t = \one $}  \\
     & \vec z \geq \mathbf{0}  \notag
  \end{align}
where the $z_i$ can be any positive real numbers.
Chan \etal \cite{chan2} consider a similar relaxed linear program
for noisy group testing.

\section{Bounds on rates}\label{sec:rates}

In this section, we give the main results of this paper. Below, we state bounds
on the maximal achievable rates (recall Definition \ref{def:cap}) of our algorithms
with Bernoulli test designs. The bounds are illustrated in Figure \ref{fig:capsketch}.

\begin{figure}
	\centering
		\includegraphics[scale=1.5]{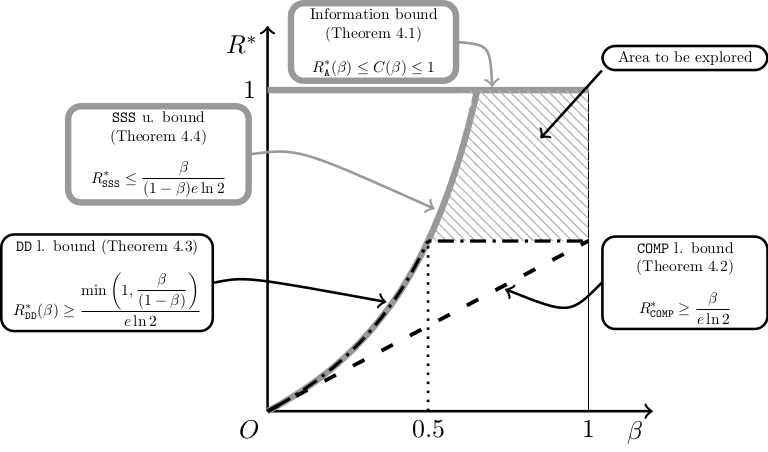}
\caption{Bounds on achievable rates of the algorithms \algCOMP, \algDD, and \algSSS\ for Bernoulli
  test designs, shown with the information bound on capacity.} \label{fig:capsketch}
\end{figure}

First, from a simple counting argument, we have the following capacity bound, which
we refer to as the information bound.

\begin{theorem} 
  For any algorithm $\mathtt A$, we have $R^*_{\mathtt A}(\beta) \leq C(\beta) \leq 1$.
\end{theorem}

For a formal proof with explicit bounds on error probability, see for example \cite{chan2011} or \cite{johnsonc10}. The paper \cite{johnsonc10} proves  
strong error bounds (see Equation (\ref{eq:infobd})), corresponding to a `strong converse' in information theory.

Second, simple manipulation of Theorem \ref{chanboundthm},
a bound  
on the error probability of \algCOMP\  due to Chan et al \cite{chan2011,chan2},  gives the following rate bound:
\begin{theorem} \label{compthm}
  For the \algCOMP\ algorithm with a Bernoulli$(1/K)$ test design, we have
    \[ R^*_{\COMP} \geq \frac{\beta}{e \ln 2} \approx 0.53 \beta . \]
\end{theorem}
 We give an alternative proof of Theorem \ref{chanboundthm} in 
Remark \ref{rem:chan2}, and prove Theorem \ref{compthm}
in Appendix \ref{ap:comprate}.

For our new \algDD\ algorithm we have the following lower bound on rate.

\begin{theorem} \label{ddthm}
  For the \algDD\ algorithm with a Bernoulli$(1/K)$ test design, we have
    \[ R^*_{\mDD} \geq \frac{1}{e \ln 2} \min \left\{ 1, \frac{\beta}{1-\beta} \right\} \approx 0.53 \min \left\{ 1, \frac{\beta}{1-\beta} \right\} . \]
\end{theorem}

In Appendix \ref{sec:ddsuc} we give an exact expression for the error probability of \algDD; in Appendix
\ref{sec:ddprob2}  we bound this expression, giving an easier-to-use approximation; and in Appendix
\ref{ap:ddrate} we convert this into the above rate bound.

Comparing Theorems \ref{compthm} and \ref{ddthm}, we see that for $0 < \beta < 1$,  the performance guarantees for \algDD\ strictly exceed those for \algCOMP.

Finally, for the \algsss\ algorithm, we have the following upper bound on rate. Since \algsss\ is essentially
optimal for Bernoulli tests, we argue that our new detection algorithms should be compared with this as the limit
of what may be possible with Bernoulli test designs.

\begin{theorem} \label{sssthm}
  For the \algSSS\ algorithm with any Bernoulli test design, we have
    \[ R^*_{\SSS} \leq \frac{1}{e \ln 2} \frac{\beta}{1-\beta} \approx 0.53 \frac{\beta}{1-\beta} . \]
\end{theorem}

In Appendix \ref{sec:sssprob2} we give a upper bound on the success probability of
the \algSSS\ algorithm, which in Appendix \ref{ap:sssrate} we convert to the
upper bound on rate of Theorem \ref{sssthm}. We also give an lower bound on the success probability of the \algSSS\ algorithm,
in Appendix \ref{sec:sssprob}.

From Theorems \ref{ddthm} and \ref{sssthm}, we see that the \algDD\ algorithm achieves
the same rate as \algSSS\ for $\beta \leq \frac12$, and hence is essentially optimal in this regime.
We note also that for
  \[ \beta \leq \frac{e \ln 2}{1 + e \ln 2} \approx 0.65 ,\]
the rate of the \algSSS\ algorithm is bounded away from the information bound $C(\beta) \leq 1$, which is achievable
for adaptive testing \cite{johnsonc10}. There are two possible explanations for this. One is that
Bernoulli tests are suboptimal in these regimes -- and very far from optimal in the denser cases.
The other is that there is an `adaptivity gap', and no nonadaptive algorithms can perform as well
adaptive algorithms here, with a gap that increases as the problem becomes denser.

Unfortunately, the complicated sequential nature of \algSCOMP\ makes it difficult to analyse
mathematically. However,
simulations in Section \ref{sec:simulation} show that in practice \algSCOMP\ performs better than 
\algDD. Hence, we conjecture that
  \[ R^*_{\mathtt{SCOMP}} \begin{cases} = \displaystyle \frac{1}{e \ln 2} \frac{\beta}{1-\beta} & \text{for $\beta \leq \frac12$,} \\
                                        \geq  \displaystyle \frac{1}{e \ln 2}  & \text{for $\beta > \frac12$.} \end{cases} \]

The proofs of these theorems is sometimes quite involved, and we save details for the appendices.
In Appendix \ref{sec:error}, we prove explicit
bounds on the error probability of \algDD\ and \algSSS. In Appendix \ref{sec:capacity} we convert the error probability
bounds into the bounds on achievable rates we see above. In Appendix \ref{sec:facts}, we summarize some elementary probability facts we will use. 
  
\section{Simulations} \label{sec:simulation}

In this section, we run simulations of our new algorithms, and compare our theoretical bounds
to empirical results.
All simulations were run with $N = 500$ items, of which $K = 10$ were defective
(except for Figure \ref{graph3}), and
Bernoulli test matrices with parameter $p = 1/K$.
Each plotted point is based on the average success rate from $1000$ simulations.

\begin{figure}
	\centering
		\includegraphics[scale=.6]{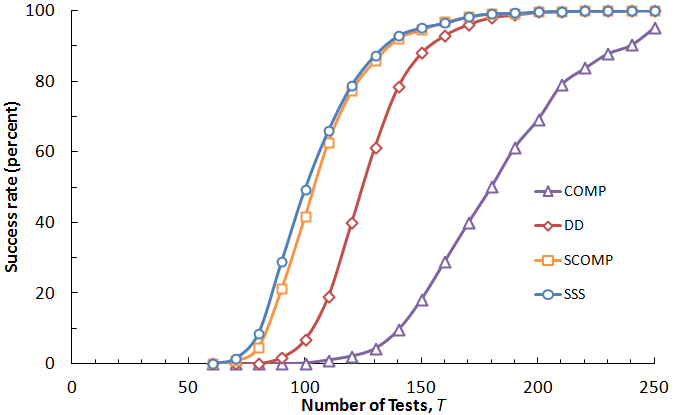}
	\caption{Performance of the \algCOMP, \algDD, \algSCOMP\ and \algSSS\ algorithms for noiseless
	  group testing with a Bernoulli test design with $N = 500$, $K = 10$, $p = 1/10$.}
	\label{graph1}
\end{figure}

Figure \ref{graph1} shows the performance of the algorithms featured in this paper.
Our \algDD\ algorithm far outperforms the \algCOMP\ algorithm of Chan \emph{et al.}, and our \algSCOMP\
algorithm is better still. For this moderately sized example ($N = 500$, $K = 10$,
$T \approx 100$), it was possible to use an integer linear programming solver to run the
\algSSS\ algorithm (even without the improvement we mention in Section \ref{sec:sss}). Promisingly, the computationally simple \algSCOMP\ algorithm has performance very close to that
of the essentially optimal but computationally hard \algSSS\ algorithm; the performance is
particularly close in the most important high success probability regime.

\begin{figure}
	\centering
		\includegraphics[scale=.6]{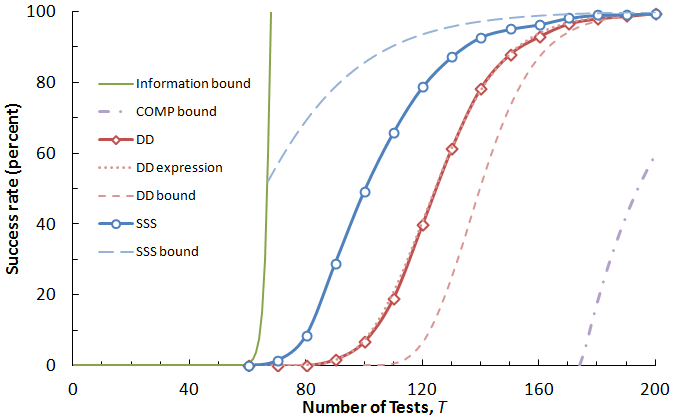}
	\caption{Performance of the \algDD\ and \algSSS\ algorithms, with the	
		information lower bound (Theorem \ref{thm:basic}), the \algCOMP\ lower bound (Corollary
		\ref{cor:chanbd})  of Chan \emph{et al.}, our exact expression for \algDD\ (Theorem \ref{thm:lower}),
		our \algDD\ lower bound (see (\ref{eq:ddsucbd}) in  Lemma \ref{lem:innersum}), and our \algSSS\ upper bound (Theorem \ref{prop:sample}) for	noiseless group testing with a
		Bernoulli test design with $N = 500$, $K = 10$, $p = 1/10$.}
	\label{graph2}
\end{figure}

Figure \ref{graph2} shows the performance of the \algDD\ algorithm. The algorithm does indeed perform as predicted analytically, and our bound is reasonably tight, especially in the 
high success
 probability regime. Note also that our bound on success probability is a big improvement on the Chan \etal bound for the \algCOMP\ algorithm. While performance of \algDD\ is far from the information theoretic bound, the essentially optimal \algSSS\ algorithm shows that the bound is very far from achievable with a Bernoulli test design and $K$ unknown.

\begin{figure}
	\centering
		\includegraphics[scale=.6]{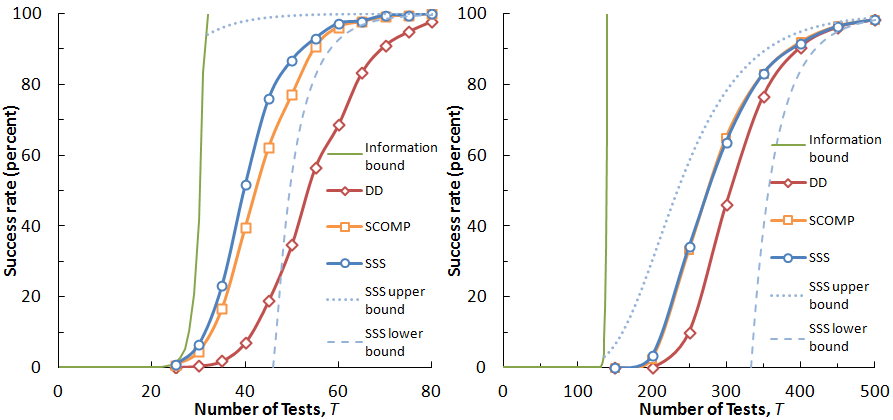}
	\caption{Performance of the \algSSS\ algorithm and associated bounds (Theorems \ref{thm:basic},
	  \ref{thmsss} and \ref{prop:sample}) and the \algSCOMP\ algorithm for noiseless group testing
	  with a Bernoulli test design with $N = 500$ in a sparse case (left: $K = 4$, $p = 1/4$,
$\beff = 0.7769)$
	  and a dense case (right: $K = 25$, $p = 1/25$, $\beff = 0.4820$.)
	\label{graph3}}
\end{figure}

Figure \ref{graph3} illustrates the difference between a sparse case (left subfigure) and dense case (right subfigure) of group testing. In the sparse case, our \algSSS\ upper bound is generally loose compared to the information bound, while the lower bound is generally right, especially in the high success rate regime. Here, the \algSCOMP\ algorithm slightly underperforms the more computationally difficult \algSSS\ algorithm. In the dense case, on the other hand, our \algSSS\ upper bound is much tighter than the information bound, while the lower bound is fairly loose away from the high success rate regime. Here, the \algSCOMP\ algorithm performs essentially equivalently to the difficult \algSSS\ algorithm, and even \algdd performs close to the \algsss.

We can understand the performance illustrated in Figure \ref{graph3} in terms of the rate
results of Appendix \ref{sec:capacity}. In particular, given $N$ and $K$, we write 
$\beff = 1 - \ln K/\ln N$ for the value such that $K = N^{1-\beff}$. The sparse case has
the value $\beff = 0.7769$, and corresponds to the region $\beta > 1/2$ where the rate bounds are less tight and the \algdd algorithm is probably not optimal. 
In contrast, the denser case
has $\beff = 0.4820$ and corresponds to the region $\beta < 1/2$ where the \algdd algorithm 
asymptotically converges to the \algsss bound.

\section{Conclusions and further work}

We have introduced several new algorithms for noiseless non-adaptive group testing, 
including the \algdd algorithm and \algscomp algorithm.
We have demonstrated by 
bounds on their maximum achievable rates and by direct simulation that they perform well compared with known algorithms
in the literature, and in some denser cases are asymptotically optimal.

We briefly mention some problems for future work:
\begin{enumerate}
\item To give asymptotic bounds on the performance of the \algsss algorithm, which would require a more detailed combinatorial analysis.
Such asymptotic bounds
 would allow us to deduce tighter bounds on the value of $C(\beta)$ for $\beta > 1/2$.
\item To compare the performance of algorithms under Bernoulli test designs and other matrix designs, including
the LDPC-inspired designs of Wadayama \cite{wadayama}.
\item To develop similar algorithms and bounds for the noisy case. 
\end{enumerate}
\appendix

\section{Proofs: bounds on error probability} \label{sec:error}

\subsection{Information bound}

For comparison, we note the information  bound
in a form due to  Baldassini, Johnson and Aldridge \cite[Theorem 3.1]{johnsonc10}:
\begin{theorem}[\cite{johnsonc10}] \label{thm:basic}
Consider testing a set of $N$ items with $K$
defectives. Any algorithm to recover the defective set with $T$
tests has success probability satisfying
  \begin{equation} \label{eq:infobd} \pr(\suc) \leq \frac{2^T}{\binom{N}{K}}. \end{equation}
\end{theorem} 

This theorem strengthened a result of Chan \etal~\cite[Theorem 1]{chan2011}, who
referred to their bound as `folklore', noting that similar bounds
appear in the literature, such as \cite{dyachkov3}. 
\subsection{\algcomp} \label{sec:compprob}

Chan \etal \cite[equation (8)]{chan2011}\cite[equation (34)]{chan2}
give the following bound on the success probability of the \algcomp algorithm: 
\begin{theorem}  \label{chanboundthm}
  For noiseless group testing with a Bernoulli$(p)$ test design, the success probability of the \algcomp
  algorithm is bounded by
\begin{equation} \pr(\suc) \geq 1 - (N-K)\big(1-p(1-p)^K\big)^T . \label{eq:comperror}
\end{equation}
\end{theorem}

By differentiation, it is easy to see that   (\ref{eq:comperror}) is maximised at
$p = 1/(K+1)$, agreeing
with Johnson and Sejdinovic's argument that $p = 1/K$ is asymptotically
optimal in the $\lim_{K\to\infty} \lim_{N\to\infty}$ regime \cite{johnsonc8}.
Note that we show in Remark \ref{rem:chan2} below
that  Theorem \ref{chanboundthm} can be recovered using our techniques.

\subsection{\algDD: exact expression} \label{sec:ddsuc}

We now derive an exact expression for the probability that the \algdd 
algorithm
recovers  the defective set exactly. It is helpful to mentally sort the rows and columns of the testing matrix $\mat X$
(and outcome vector $\vec y$) in a suitable way; this implies no loss of generality. This
is illustrated in Figure \ref{fig:matrix}.

\begin{figure}
\begin{center}
\includegraphics{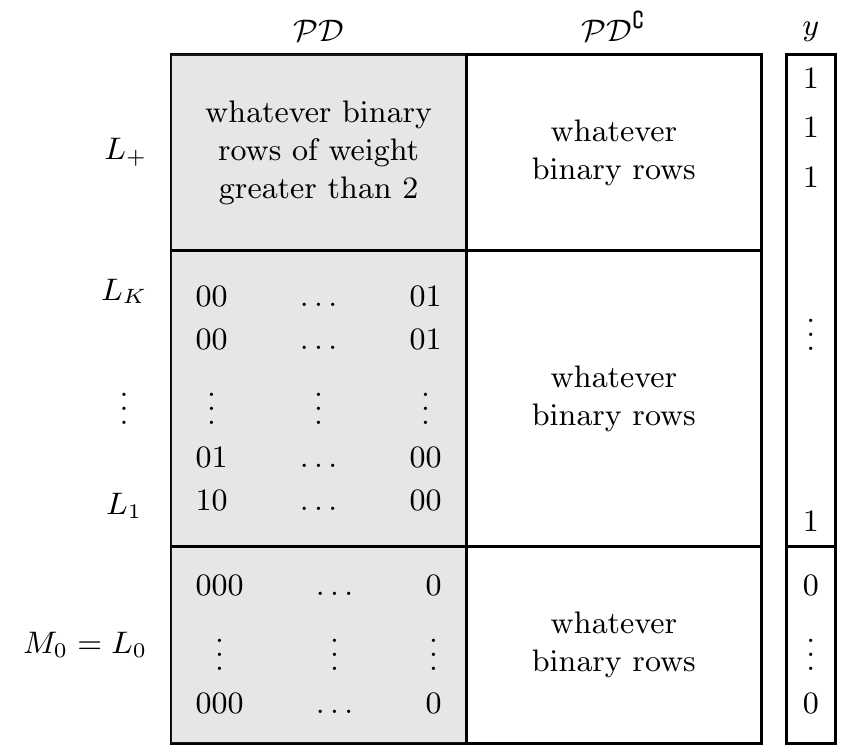}
\caption{The testing matrix $\mat X$, where the rows have been grouped 
according to the partition $\vc{L}$ and the columns into the 
two subsets $\pd$ and $\pd^\compl$. 
The shaded area is the submatrix $\subm$. Notice that the variables 
$L_0,\ldots, L_+$ do not index single rows but groups of rows, according to the 
definition of $\vc{L}$ \eqref{ldef}. }
\label{fig:matrix}
\end{center}
\end{figure}

 Recall from Definition \ref{def:nd} that we write $\pd = \DND^\compl$ for the set
of possible defectives (items which do not appear in any negative test).
It will also be convenient here to, without loss of generality, label the
actual $K$ defectives as $\K = \{ 1,2,\dots, K\}$.
Note that $\K \subseteq \pd$, and the only type of error we can make
is a false negative, when a defective item is masked.
It will be useful to consider the following partition of the number of tests $T$, which depends on the random matrix $\mat X$ and the defective set $\K$:
  \begin{align}
    L_0 &= \text{$\#$ tests with no defective items in,} \notag \\  
L_i &= \text{$\#$ tests containing $i$ and no other element of  $\pd$, for $i=1, \ldots, K$,} \label{ldef} \\
    L_+ &= \text{$\#$  remaining tests.} \notag
  \end{align}
If $L_i \neq 0$ for some $i \in \K$ then the \algdd algorithm correctly identifies the defective element
$i$. Thus the success probability of the \algdd algorithm is precisely the
probability that $L_i \neq 0$ for all $i \in \K$
 \begin{equation} \label{eqsuc}
   \pr(\suc) = \pr \left( L_1 \neq 0, \dots, L_K \neq 0 \right) .
 \end{equation}
For this reason,
we want to know the distribution of $\vec L = (L_0, L_1, \dots, L_K, L_+)$.
Unfortunately the distribution of $\vec L$ is quite complicated, but
we will be able to analyse it conditioned on the number of possible
defectives $|\pd| = K +G$ and a related random vector $\vec M$.
Recall from Section \ref{sec:dd} that $G$ is the number of nondefective
items that are nonetheless in $\pd$.
We define $\vec M = (M_0, M_1, \dots, M_K, M_+)$ as follows:
  \begin{align*}
    M_0 &= \text{$\#$ tests with no defective items in,} \\
    M_i &= \text{$\#$ tests with $i$ the only defective item in, for $i=1, \ldots, K$} \\
       M_+ &= \text{$\#$ tests with two or more defective items in.}
  \end{align*}
Note that this is similar to the definition of $\vec L$, but with the
set of possible defectives $\pd$ replaced by the set of actual defectives
$\K$.
Write
  \begin{align*}
    q_0 &= \pr(\text{no defectives}) = (1-p)^K , \\
    q_1 &= \pr(\text{$1$ the only defective}) = p(1-p)^{K-1} , \\
    q_+ &= \pr(\text{two or more defectives}) = 1 - q_0 - Kq_1 , \\
\vc{q} & =(q_0, \underbrace{q_1, \dots, q_1}_{\text{$K$ terms}}, q_+). 
  \end{align*}
The following lemma then holds:
\begin{lemma} \label{lem:sampling}
Using a  {\rm Bernoulli}$(p)$ test design,  $\vc{M} = (M_0, M_1, \dots, M_K, M_+)$  is multinomial with probability distribution $\pr_{T; \vc{q}}$ defined as
	\begin{equation} \label{eqM}
		\pr_{T; \vc{q}}(m_0, m_1, \ldots, m_k, m_+)=\df{T!}{m_+!\prod_{i=0}^K m_i!}q_0^{m_0} q_1^{m_1+\cdots+m_k} q_+^{m_+}\ ,
		\end{equation}
for $\sum_{i=0}^K m_i+m_+=T$.
\end{lemma}

By analysing the relationship between $\vc{M}$ and $\vc{L}$, we are able to give the 
probability of success of
the \algdd algorithm. The strategy is as follows: from \eqref{eqM} we have the distribution
of $\vec M$; given $\vec M$,  \eqref{eqS} below gives the distribution
of $G$; given $\vec M$ and $G$, \eqref{eqL} gives the distribution
of $L$; and given the distribution of $L$, we have from
\eqref{eqsuc} the probability of success.

Putting this all together, we can derive an exact expression for the success probability
of the \algdd algorithm,
in terms of the binomial mass function $\binomi{n}{t}{k} := \binom{n}{k} t^k(1-t)^{n-k}$ and
$\phi_K$, defined as
\begin{equation}
\label{eq:phi}
	\phi_K(q,T) := \sum_{\ell = 0}^K (-1)^\ell \binom{K}{\ell} (1 - \ell q)^T\ ,
\mbox{\;\;\; for  $q\in[0, 1]$.}
\end{equation}
Appendix \ref{sec:facts} summarises some well-known results from probability theory, including properties of the multinomial distribution. In particular,
Lemma \ref{lem:maskprob} shows that $\phi_K$ gives the probability of a
set of components of 
a multinomial being non-zero, in a certain symmetric situation.
\begin{theorem}\label{thm:lower}
Given a {\rm Bernoulli}$(p)$ test design, the success probability of the 
\algdd algorithm is 
\begin{equation} \label{eq:ddsuc}
\pr(\suc) =\sum_{m_0 = 0}^{T} \sum_{g = 0}^{N-K} \binomi{T}{q_0}{m_0} 
\binomi{N-K}{ (1-p)^{m_0}}{g} \Phi_K(g,m_0)\ , \end{equation}
where we write
$ \Phi_K(g,m_0) = \phi_K \left(q^*(g), T-m_0\right)$ for
 $q^*(g) = q_1 (1-p)^g/(1-q_0)$ . 
\end{theorem}
\begin{proof} The key is to condition on the value of $M_0$. By Lemma
\ref{lem:multifacts}.\ref{it:marginal}, 
$M_0 \sim \bino( T, q_0)$, and by Lemma \ref{lem:multifacts}.\ref{it:cond}, conditioned on $M_0 = m_0$, the vector $\vc{M}' =
(M_1, \ldots, M_k, M_+) \sim \pr_{T-m_0; \vc{q}'}$, where
$$ \vc{q}' = \left( \frac{q_1}{1-q_0}, \ldots, 
\frac{q_1}{1-q_0}, \frac{q_+}{1-q_0} \right).$$
Next, given $M_0$, we can find the distribution of $G$, the number of intruding non-defectives.
First, all $K$ actual defectives will be in $\pd$. Then each of the other $N-K$
items will fail to be in any of the $M_0$ negative tests with probability
$(1-p)^{M_0}$. Hence we have that, independent of $\vc{M}'$,  the conditional distribution
of $G$ given $M_0$ is
  \begin{equation} \label{eqS}
    G \mid M_0 \sim \text{Bin}(N-K,(1-p)^{M_0}) .
  \end{equation}
We will express the 
success probability as 
\begin{equation} \label{eq:succ}
\pr(\suc) = 
 \sum_{m_0 = 0}^{T} \sum_{g = 0}^{N-K} \binomi{T}{q_0}{m_0} \binomi{N-K}{(1-p)^{m_0}}{g} 
\pr( \suc \mid M_0 = m_0, G = g ). \end{equation}
Next, given $\vec M$ and $G$, we can write down the conditional
distribution of $\vec L$, since for $i=1, \ldots, K$,
  \begin{align}
    L_0 &= M_0, \notag \\
    L_i \mid M_i, G &\sim \text{Bin}(M_i, (1-p)^{G}), \label{eqL} \\
    L_+ &= T - \sum_{i=0}^K L_i .\notag 
  \end{align}
This is because for each $i$, a test which contains defective item $i$ and no other defective will
contribute to $M_i$. However, it will only contribute to $L_i$ if it also contains none
of the $G$ intruding non-defectives. The Bernoulli sampling of the matrix $\mat{X}$ means that each such test
will contribute to $L_i$ independently with probability $(1-p)^G$. Equivalently, the $L_i$ are independently thinned versions of $M_i$ (in the
sense of R\'{e}nyi \cite{renyi4}), with
thinning parameter $(1-p)^{G}$.

Repeatedly using Lemma \ref{lem:multifacts}.\ref{it:split}, we can
deduce that, conditional on $M_0 = m_0$ and $G=g$, we have that
\[ (L_1, L_2, \ldots, L_k, M_1-L_1, \ldots, M_k-L_k, M_+)
\sim \pr_{T-m_0, \vc{q}''}, \]
where
\[
\vc{q}'' = \Big( \frac{q_1 (1-p)^g}{1-q_0}, \ldots, 
\frac{q_1 (1-p)^g}{1-q_0}, 
\frac{q_1 (1-(1-p)^g)}{1-q_0}, \ldots, 
\frac{q_1(1-(1-p)^g)}{1-q_0},
\frac{q_+}{1-q_0} \Big).
\]
 From  (\ref{eqsuc}), we know that the \algdd
algorithm will be successful
precisely in the event $\bigcap_{i=1}^K \{ L_i \neq 0 \}$.
Lemma \ref{lem:maskprob} gives an exact expression for this probability as 
$$\pr( \suc | M_0 = m_0, G = g ) = \phi_K  \left( \frac{q_1 (1-p)^g}{1-q_0},
T-m_0 \right).$$
We can then directly substitute this in  (\ref{eq:succ}) to deduce the theorem.
\end{proof}

\begin{remark} \label{rem:chan2}
We can recover the bound (\ref{eq:comperror}) of Chan \etal \cite{chan2011}  using our techniques. As previously
remarked in Section \ref{sec:dd}, their \algcomp algorithm succeeds if and only if $G = 0$.
Using  (\ref{eqS}) we know that  $G \mid M_0 \sim \text{Bin}(N-K,(1-p)^{M_0})$.
This means that
\begin{align}
\pr( \suc) & = \sum_{m_0 = 0}^T \pr(M_0 = m_0) \pr(G = 0 |M_0 = m_0)  \notag \\
& =  \sum_{m_0=0}^T \binom{T}{m_0} q_0^{m_0} (1-q_0)^{T-m_0} \left(  (1- (1-p)^{m_0} )^{N-k}
\right) \label{eq:bracketedterm} \\
& \geq  \sum_{m_0=0}^T \binom{T}{m_0} q_0^{m_0} (1-q_0)^{T-m_0} \left( 1 - (N-k) (1-p)^{m_0} \right) \notag \\
& = 1-  (N-k) \left( q_0 (1-p) + 1-q_0) \right)^T. \notag
\end{align}
Here we bound the bracketed term (\ref{eq:bracketedterm}) using the Bernoulli inequality 
\begin{equation} \label{eq:bern} (1+x)^T \geq 1 + x T  \mbox{ for all
$x \geq -1$ and $T \geq 0$,} \end{equation} and the result follows since $q_0 = (1-p)^k$, so that
$q_0(1-p) + 1 - q_0 = 1 - p(1-p)^k$.
\end{remark}

\subsection{\algDD: bounds} \label{sec:ddprob2}

Theorem \ref{thm:lower} gives a complicated multipart expression that gives the success  probability of the 
\algdd algorithm. In fact, since $\Phi_K$ is defined in terms of a summation formula, 
the expression (\ref{eq:ddsuc}) is a triple sum, which is difficult to analyse and control.

	Notice that for the success probability we can use the bound $\phi_K( q, T-m_0) \geq \max\{0,1 - K( 1- q)^{T-m_0}\}$ (see Lemma \ref{lem:maskprob}) to reduce the equality \eqref{eq:succ} to a lower bound, expressed in terms of a double sum. 
In this subsection we derive a simpler bound on the success probability. We repeatedly use the fact that
\begin{equation} \label{eq:expbound}  (1-x)^y \leq \exp(-x y) \mbox{\;\;\; for $0 \leq x \leq 1$ and $y \geq 0$.} \end{equation}
In order to analyse the success probability of the \algdd algorithm, it is useful to bound $q_0$ and $q_1$. 
\begin{lemma} \label{lem:qbounds} Writing $p = 1/K$, and defining $q_0 = (1-p)^K$ and $q_1 = p(1-p)^{K-1}$,  we  notice that $q_1 = q_0/(K-1)$. Hence for
 $K \geq 2$ we deduce that:
\begin{align}
 \frac{K-2}{K-1} e^{-1} & \leq  q_0  \leq e^{-1}, \label{eq:q0bd} \\
\frac{q_1}{1-q_0} & \leq  \frac{1}{K}. \label{eq:q0q1bd}
\end{align}
\end{lemma}
\begin{proof}
The upper bound of (\ref{eq:q0bd}) follows by taking $x = p$ and $y= K$ in (\ref{eq:expbound}) above. The lower bound is slightly more involved; taking $x= -p/(1-p)$ and $y=K$,
we deduce from (\ref{eq:expbound}) that $(1 - p)^K \geq \exp(-1 - \frac{1}{K-1})$. Further, taking
$x =p/(1-p)$ and $y=1$ in (\ref{eq:expbound}) tells us that $\exp(-\frac{1}{K-1}) \geq \frac{K-2}{K-1}$, and the result follows.

Further, Equation (\ref{eq:q0q1bd}) follows  on rearranging the fact that $q_0 + K q_1 \leq 1$.
\end{proof}

We bound $\Phi_K$, using arguments based on  the Bernoulli inequality (as previously used in Remark \ref{rem:chan2}).
\begin{lemma}
\label{lem:bound}
For fixed $T$, $K$ and $n_0$,
 taking
  \[ q^*(g):=\frac{q_1 (1-p)^g}{1-q_0}=\frac{p(1-p)^{K+g-1}}{1-(1-p)^K} , \]
the function 
\begin{eqnarray}
 \Phi_K(g,m_0) & = &  \phi_K( q^*(g), T-m_0) \nonumber \\
& \geq & \max \left\{0, 1- K \exp \left( - \frac{ q_1(T-m_0)}{1-q_0} \right) \exp
\left( \frac{ p q_1(T-m_0)}{1-q_0} g \right) \right\}. \label{eq:errbound}
\end{eqnarray}
\end{lemma}
\begin{proof} First observe that
as in Lemma \ref{lem:maskprob}  we can write
\begin{equation} \label{eq:PhiK} 1- \Phi_K(g,m_0) \leq K( 1 - q^*(g) )^{T-m_0} \leq K \exp \left( - (T-m_0) q^*(g) \right),\end{equation}
since (\ref{eq:expbound}) gives that
$(1-x)^y \leq \exp(-x y)$ for $0 \leq x \leq 1$ and $y \geq 0$.
Further, the Bernoulli inequality (\ref{eq:bern}), giving $(1+x)^T \geq 1 + x T$, means that
$$ q^*(g) = \frac{q_1 (1-p)^g}{1-q_0} \geq \frac{ q_1(1- p g)}{1-q_0},$$
and substituting this  in (\ref{eq:PhiK}), the result follows.
\end{proof}

We now prove a theorem to bound the success probability of the \algdd algorithm, by
exploiting the favourable geometry of the distributions of $(M_0, G)$ and $\Phi_K$, the
probability 
that no defectives are masked. The strategy is that since $M_0$, the number of tests containing no defectives, is concentrated around
its mean $T q_0$, then bounding $\Phi_K$ will give a bound on the overall success probability, by controlling the inner sum in Theorem \ref{thm:lower}.
\begin{lemma} \label{lem:innersum} For any given $m_0$ we can bound the inner sum in Theorem \ref{thm:lower} by
\begin{equation}
\sum_{g = 0}^{N-K} \binomi{N-K} {(1-p)^{m_0}}{g} \Phi_K(g,m_0)
\geq  \max \left\{ 0, 1- K \exp \Theta(T,m_0) \right\},
\end{equation}
where we write
\begin{equation} \label{eq:Thetadef} \Theta(T,m_0) = N (1-p)^{m_0}  \left( \exp
\left( \frac{(T-m_0) p q_1}{1-q_0} \right) - 1 \right) - \frac{q_1 (T-m_0)}{1-q_0}. \end{equation}
Hence, given a {\rm Bernoulli}$(p)$ test design, the probability of success under the 
\algdd algorithm satisfies
\begin{equation} \label{eq:ddsucbd}
\pr(\suc) \geq \sum_{m_0 = 0}^{T}  \binomi{T}{q_0}{m_0} 
\max \left[ 0, 1- K 
\exp \left( \Theta(T,m_0) \right) \right]. \end{equation}
\end{lemma}
\begin{proof}
Using Lemma \ref{lem:bound}
we can simply bound the left-hand side by writing $\ol{p} = (1-p)^{m_0}$ and using the binomial theorem:
\begin{align}
&\sum_{g = 0}^{N-K} \binom{N-K}{g} \ol{p}^g (1-\ol{p})^{N-K-g} \left[
 K \exp \left( - \frac{ q_1(T-m_0)}{1-q_0} \right) \exp
\left( \frac{(T-m_0) g p q_1}{1-q_0} \right) \right] \notag \\
&\qquad =  K \exp \left( - \frac{ q_1(T-m_0)}{1-q_0} \right)
\sum_{g = 0}^{N-K} \binom{N-K}{g} \left( \ol{p} \exp \left( \frac{(T-m_0) p q_1}{1-q_0} \right) \right)^g (1-\ol{p})^{N-K-g} \notag \\
&\qquad =  K \exp \left( - \frac{ q_1(T-m_0)}{1-q_0} \right)
\left( \ol{p} \exp \left( \frac{(T-m_0) p q_1}{1-q_0} \right) + 1-\ol{p} \right)^{N-k} \notag \\
&\qquad \leq K \exp 
\left( (N-K) (1-p)^{m_0} \left( \exp \left( \frac{(T-m_0) p q_1}{1-q_0} \right) - 1 \right) - \frac{ q_1(T-m_0)}{1-q_0} \right). \label{eq:term}
\end{align}
where the final inequality follows using $(1+x)^y \leq \exp(x y)$ for positive $x$ and $y$.
\end{proof}

\subsection{\algSSS: lower bound} \label{sec:sssprob}

We have the following lower bound on the success probability of the \algsss algorithm.
\begin{theorem}\label{thmsss}
  For noiseless group testing with a Bernoulli($p$) test design, the success probability of  the \algsss
  algorithm is bounded by
  \begin{equation} \label{eq:sssbound}
     \pr(\suc) \geq 1 - K \big(1 - Q(K,K-1,K-1)\big)^T
                      - \sum_{B=0}^{K-1} \binom{K}{B} \binom{N-K}{K-B} \big(1 - Q(K,K,B)\big)^T ,
  \end{equation}
  where we write
    \begin{equation} \label{Q}
      Q(K,L,B) = (1-p)^K + (1-p)^L - 2(1-p)^{K+L-B} .
    \end{equation}
\end{theorem}

The final term in \eqref{eq:sssbound}, which corresponds to the error probability
when $K$ is known, has previously been analysed by Seb\H{o} \cite{Sebo198523} in the fixed $K$ regime
(equivalent to our $\beta = 1$).

To prove Theorem \ref{thmsss} we will require the following lemma.

\begin{lemma} \label{lem}
  The probability that a single Bernoulli($p$) test $\vec x$ gives a
  different outcome depending on whether the defective set is $\K$ or
  $\L$ is $Q(|\K|,|\L|,|\K\cap\L|)$, where $Q$ is as in \eqref{Q}.
\end{lemma}

\begin{proof}
  Write $K = |\K|$, $L = |\L|$, and $B = |\K\cap\L|$ for respectively the number of
  items in $\K$, in $\L$, and in both $\K$ and $\L$. Also write $q = 1-p$.

  The test gives a negative outcome with defective set $\K$
  but a positive outcome with defective set $\L$ if and only
  if no items of $\K$ are included in the test, but at least one item
  $\L \setminus \K$ is included.  This occurs with probability
    \[ q^K(1-q^{L-B}) = q^K - q^{K+L-B} . \]
  Similarly, the test gives a positive outcome with defective set $\K$
  but a negative outcome with defective set $\L$ with probability
    \[ q^L(1-q^{K-B}) = q^L - q^{K+L-B} . \]

  Adding together the probabilities of these disjoint events gives
  the result.
\end{proof}

We can now prove Theorem \ref{thmsss}.

\begin{proof}
The \algsss algorithm may make an error if the true defective set
$\K$ is not the unique smallest satisfying set.  Thus the error probability of \algsss
is
  \[ \epsilon \leq \pr \Bigg( \bigcup_{\substack{|\L| \leq K\\ \L \neq \K}}
            A(\L,\K) \Bigg) , \]
where $A(\L, \K)$ denotes the event that the sets $\L$ and $\K$ would give identical
outcomes for all $T$ tests.

Consider a set $\L$ containing $|\L| = L$ items, where there are $B = |\K \cap \L|$ items
in both $\L$ and $\K$. By Lemma \ref{lem} and the fact that tests are independent, we have that
  \[ \pr \big( A(\L, \K) \big) = \big(1 - Q(K,L,B)\big)^T \]

At this stage we can get a simple bound by using the union bound to write
  \begin{align}
    \epsilon &\leq \sum_{\substack{|\L| \leq K\\ \L \neq \K}} \pr \big( A(\L, \K) \big) \notag \\
             &= \sum_{\substack{|\L| \leq K\\ \L \neq \K}} \big(1 - Q(K,|\L|,|\K \cap \L|)\big)^T \notag \\
             &= \sum_{L=0}^{K} \sum_{B=0}^{L} \binom{K}{B} \binom{N-K}{L-B} \label{union}
\big(1 - Q(K,L,B)\big)^T - 1,
  \end{align}
where the $-1$ is because our sum includes a term for $L = B = K$, corresponding to the
true defective set.

However, we can get a tighter bound by noting that many of the events $A(\K,\L)$ are subsets
of other events of the same type. First, given a satisfying set $\L$ with $B = L < K-1$ -- that is, with no false
positives and at least two false negatives -- the event $A(\K,\L) \subset A(\K, \L')$ with
$L' = B' = K - 1$, where $\L'$ is the set $\L$ with extra defective items added.
Second, given a satisfying set $\L$ with $B < L < K$ -- that is, with at least one false positive and at least
one false negative -- the event $A(\K,\L) \subset A(\K, \L'')$ with $L'' = K$, where again
$\L''$ is the set $\L$ with extra defective items added.

Considering only the terms in \eqref{union} with $L = B = K-1$ and the terms
with $L = K$ gives  the tighter bound
  \[ \epsilon \leq K \big(1 - Q(K,K-1,K-1)\big)^T
                     + \sum_{B=0}^{K-1} \binom{K}{B} \binom{N-K}{K-B} \big(1 - Q(K,K,B)\big)^T \]
as desired.
\end{proof}

\subsection{\algSSS: upper bound} \label{sec:sssprob2}

Next, in Theorem \ref{prop:sample},
we give an upper bound on the
success probability of the \algsss algorithm. As  discussed in Section \ref{sec:sss}, this algorithm can be 
viewed as an idealized benchmark for the performance of any algorithm (when the number of
defectives is unknown), so this bound should control the success probability of any algorithm.
\begin{theorem} \label{prop:sample}
For any $K$, if  we sample $\mat X$ according to a Bernoulli($p$) test design, then for any $p$,
the \algsss algorithm has success probability bounded above by
\begin{equation} \pr( \suc) \leq \phi_K \left(  \frac{1}{e(K-1)} , T \right), \label{eq:bdsample} \end{equation}
\end{theorem}
\begin{proof}
The key is to observe that if one of the defective items is masked by the other $K-1$ defectives,
then the \algsss algorithm will not succeed, since the $K-1$ items in question form a smaller
satisfying set. 

Equivalently, the set of matrices for which \algsss succeeds is a subset of the matrices 
for which $M_i \neq 0$ for each defective $i$.
This means that for any $p$, we can write
\begin{eqnarray*}
\pr( \suc) & \leq & \pr \left( \bigcap_{i=1}^K \{ M_i \neq 0 \} \right) 
 = \phi_K( p(1-p)^{K-1}, T),
\end{eqnarray*}
where the equality follows from Lemma \ref{lem:maskprob} below.
Now, Lemma \ref{lem:nonincr} below shows that $\phi_K(q, T)$ is increasing in $q$.
Observe that since $p(1-p)^{K-1}$ is maximised at $p = 1/K$,  for any $p$, (\ref{eq:q0bd}) means
we can write $p(1-p)^{K-1} \leq  \frac{1}{K-1} e^{-1}$.
\end{proof}

It is interesting to note that the upper bound of Theorem \ref{prop:sample} is complementary
to the universal upper bound of Theorem \ref{thm:basic}. In particular, note that  (\ref{eq:bdsample})
does not depend on $N$, but only $K$. This means that which bound is tighter for a 
particular $K$ will depend on the overall sparsity of the problem.

\section{Proofs: bounds on achievable rates} \label{sec:capacity}

\subsection{A lemma for rate calculations}

In order to carry out rate calculations, it is useful to have the following limit for normalized binomial
coefficients:

\begin{lemma} \label{lem:binom}
If $K = N^{1-\beta}$ then we can write
\begin{equation} \label{eq:binomasym}
\lim_{N \rightarrow \infty} \frac{ \log_2 \binom{N}{K}}{K \ln N} = \frac{\beta}{\ln 2}.
\end{equation}
\end{lemma}
\begin{proof} Well-known bounds on the binomial coefficients (see for example \cite[Page 1097]{leiserson}) state that for any $K$, we have
\begin{equation}
\left( \frac{N}{K} \right)^K \leq \binom{N}{K} \leq \left( \frac{N e}{K} \right)^K.
\end{equation}
Taking logarithms to base 2 and dividing by $K \ln N$, using the fact that $N/K = N^\beta$, we obtain
$$
\frac{ K \beta \log_2 N}{K \ln N} \leq  \frac{ \log_2 \binom{N}{K}}{K \ln N}
\leq \frac{ K \beta \log_2 N}{K \ln N} + \frac{\log_2 e}{\ln N},
$$ 
and the result follows on sending $N \to \infty$.
\end{proof}

\subsection{\algCOMP} \label{ap:comprate}

Our new results can be contrasted with the following lower bound of Chan \emph{et al.} \cite[Theorem 4]{chan2011},
which follows by rearranging the bound on success probability in Theorem \ref{chanboundthm}:
\begin{corollary} \label{cor:chanbd}
For any $\delta > 0 $, using $T = e (1+ \delta) K \ln N$ tests  ensures that \algcomp has
$$ \pr(\suc)  \geq 1 - N^{-\delta}.$$ 
\end{corollary}
Hence we have
\begin{equation} \label{eq:chancap}
 R^*_{\mathtt{COMP}}(\beta) \geq \frac{\beta }{ e \ln 2} ,
\end{equation}
which was Theorem \ref{compthm} above.

\subsection{\algDD} \label{ap:ddrate}


\begin{theorem} \label{thm:asymdd}
Write $k(\beta) = \max\{\beta,1-\beta\}$ and fix $\delta > 0$.  Choosing
 $T = (k (\beta) + \delta) e K \ln N$ ensures that the success probability of the \algdd algorithm tends to 1 in the regime 
where $K = N^{1-\beta}$.
\end{theorem}
\begin{proof}
First, we deduce that the quantity $\Theta(T,m_0)$ defined in (\ref{eq:Thetadef}) can be bounded by the product of two terms. That is, for all $m_0$:
\begin{align}
\Theta(T,m_0)
& \leq  
N  \exp(-p m_0) \left( \exp \left( \frac{(T-m_0) p q_1}{1-q_0} \right) - 1 \right) - \frac{ q_1(T-m_0)}{1-q_0}  \notag \\
& \leq \biggl( \frac{ q_1(T-m_0)}{1-q_0} \biggr) \biggl( N p  \exp(-p m_0) \exp \left( \frac{(T-m_0) p q_1}{1-q_0} \right)   -1 \biggr)  \label{eq:term2}
\end{align}
Here, again, the first inequality uses the fact that (\ref{eq:expbound}) gives $(1-x)^y \leq \exp(-x y)$ for $0 \leq x \leq 1$ and $y \geq 0$, and
the last inequality follows using the fact that 
$\exp(x) - 1 \leq x \exp(x)$ for all $x$, and taking $x = (T-m_0) p q_1/(1-q_0)$. 

For fixed $\epsilon := \delta/6(\delta + k(\beta))$, we will separately bound the two bracketed terms of Equation (\ref{eq:term2}) for
$m_0 \in \left( T(q_0 - \epsilon/e), T(q_0 +\epsilon/e) \right)$  in Equations (\ref{eq:term4}) and (\ref{eq:term3}) below.

We control the first term of (\ref{eq:term2}), by bounding $m_0$ from  below  by $T(q_0 - \epsilon/e)$ to deduce that
(since $q_1 = q_0/(K-1)$ as in  Lemma \ref{lem:qbounds} above), for $K \geq 2$:
\begin{align} 
\frac{q_1(T-m_0)}{1-q_0} & \leq \frac{T q_0}{K-1} \left( 1 + \frac{\epsilon}{e(1 -  q_0)} \right)
\leq \frac{T}{e(K-1)} \left( \frac{K-2}{K-1} + \frac{\epsilon}{e- 1} \right) \notag \\ 
& = \ln N  (k(\beta)  + \delta) \left( \frac{K}{K-1} \left( \frac{K-2}{K-1} + \frac{\epsilon}{e- 1} \right)  \right)  \notag \\
& \leq \ln N  (k(\beta)  + \delta) \left( 1 + 2 \epsilon \right) \notag \\
& = \ln N
( k(\beta) + 4 \delta/3), \label{eq:term4} \end{align}
using the facts that (see  (\ref{eq:q0bd})), $ \frac{K-2}{K-1} e^{-1} \leq  q_0  \leq e^{-1}$, that function $t/(1-t)$ is increasing in $t$, and by the choice of $\epsilon$ given above.
Similarly, by bounding $m_0$ from below by $T(q_0 - \epsilon/e)$, we can express 
\begin{align}
- p m_0 +  \frac{(T-m_0) p q_1}{1-q_0}   & \leq - p T (q_0 - q_1) +
\frac{\epsilon p T}{e} \left( 1 + \frac{q_1}{1-q_0} \right) \notag \\
& \leq - p T q_0 \frac{K-2}{K-1} +
\frac{2 \epsilon p T}{e} \notag \\
& \leq \left( \frac{ p T}{e} \right) \left( - \left( \frac{K-2}{K-1} \right)^2 + 2 \epsilon \right) \notag \\
& \leq \ln N(  k(\beta) + \delta)(-1 + 3 \epsilon) \mbox{ \;\;\;\;\; for $K$ sufficiently large } \notag \\
& = \ln N( -k(\beta) - \delta/2),  \label{eq:term5}
\end{align}
where   we use the fact that (\ref{eq:q0q1bd}) gives
$\frac{q_1}{1-q_0} \leq \frac{1}{K} \leq 1$, the bound $ \frac{K-2}{K-1} e^{-1} \leq  q_0$ from (\ref{eq:q0bd}), and the choice of $\epsilon$
given above.
Hence, using the fact that $N p = N/K = N^{\beta}$, (\ref{eq:term5}) gives that the second term of (\ref{eq:term2}) can be 
bounded by
\begin{align}
\biggl( N p  \exp(-p m_0) \exp \left( \frac{(T-m_0) p q_1}{1-q_0} \right)  - 1 \biggr)
& \leq  N^{\beta - k(\beta)- \delta/2} - 1 \leq N^{-\delta/2} - 1, \label{eq:term3}
\end{align}
 since $\beta \leq k (\beta)$. 
Hence, for $N$ sufficiently large
and $m_0$ in this range, multiplying (\ref{eq:term4}) and (\ref{eq:term3}) gives
$\Theta(T, m_0) \leq \ln N ( - k(\beta) - \delta)$. This means that (since $k(\beta) > 1-\beta$)
we can write
$$ K \exp( \Theta(T, m_0) ) \leq N^{-\delta}.$$
Using Lemma \ref{lem:innersum}, we  deduce that the success probability satisfies
\begin{align*}
\pr(\suc) & \geq \sum_{m_0 = 0}^{T}  \binomi{T}{q_0}{m_0} 
\max \left[ 0, 1- K  \exp \left( \Theta(T,m_0) \right) \right] \\
& \geq \sum_{m_0 = T(q_0-\epsilon)}^{T(q_0 + \epsilon)}  \binomi{T}{q_0}{m_0} 
\left[ 1- K  \exp \left( \Theta(T,m_0) \right) \right] \\
& \geq  \pr \biggl( T(q_0  - \epsilon/e) \leq  M_0 \leq T(q_0 + \epsilon/e) \biggr) (1- N^{-\delta}).
\end{align*}
which converges to 1 by Chernoff's inequality, Theorem \ref{thm:chernoff}.
\end{proof}


We can now prove Theorem \ref{ddthm}, that
  \[ R_{\mathtt{DD}}(\beta) \geq \frac{1 }{e \ln 2} \min \left( 1, \frac{\beta}{1-\beta} \right). \]

\begin{proof}[Proof of Theorem \ref{ddthm}] 
Theorem \ref{thm:asymdd} shows that for  $K = N^{1-\beta}$, taking $T = ( k (\beta) + \delta) e K \ln N$ gives error probability tending to $0$, and
using the binomial coefficient bounds Lemma \ref{lem:binom} we obtain
$$ \liminf_{N \rightarrow \infty} \frac{ \log_2 \binom{N}{K} }{T}
\geq \frac{\beta}{( k(\beta) + \delta) e \ln 2},$$
which implies the desired bound.
\end{proof}

\subsection{\algSSS} \label{ap:sssrate}

We can analyse the \algsss upper bound Theorem \ref{prop:sample};
 there is a phase transition for the (appropriately normalized)
number of tests required to control the quantity $\phi_K \left(  \frac{1}{e(K-1)} , T \right)$ arising
in  (\ref{eq:bdsample}) above. That is, Theorem \ref{prop:sample} gives an upper bound on the success 
probability which roughly speaking (a) is close to 1 for more than $e K \ln K = (1-\beta)e K \ln N$ tests (b) is bounded away from 1 for 
fewer than $e K \ln K$ tests.
\begin{lemma} \label{lem:asymp}
\mbox{ } 
\begin{enumerate}
\item If for some $\delta' > 0$, we have $T \geq e(1+\delta') (K-1) \ln K$, then
$\phi_K \left(  \frac{1}{e(K-1)} , T \right) \geq 1 - K^{-\delta'}$.
\item If we have $T \leq (e (K-1) - 1) \ln K$, then
$\phi_K \left(  \frac{1}{e(K-1)} , T \right) \leq 2/3$, for any $K \geq 3$.
\end{enumerate}
\end{lemma}
\begin{proof} The key to both parts of this proof are bounds on $\phi_K$ stated as Equation (\ref{eq:bonferroni}) below, which implies
$$
1 - K (1-  p)^M  \leq \phi_K(p,M) \leq 1 - K (1-  p)^M + \frac{K^2}{2} (1-2p)^M.
$$
\begin{enumerate}
\item
Using the lower bound on $\phi_K$ stated in  (\ref{eq:bonferroni}), we know that for any $K$ and $T$,
\begin{eqnarray*}
\phi_K  \left( \frac{1}{e(K-1)}, T \right) & \geq & 1 - K \left( 1 - \frac{1}{e(K-1)} \right)^T  \\
& \geq & 1 - K \exp \left( - \frac{T}{e(K-1)} \right),
\end{eqnarray*}
so that choosing $T \geq e (1 + \delta')(K-1) \ln K$ gives that this bound is at least
$1 - K^{-\delta'}$, as required.
\item
Recall that (\ref{eq:expbound}) gives $(1-x)^y \leq \exp(-x y)$. Taking $x = -q/(1-q)$ and $y = T$, we deduce $(1-q)^T \geq \exp \left(  - \frac{q T}{1-q} \right)$.
Similarly, taking $x = q/(1-q)$ and $y = T$ gives that $   \left( \frac{1-2q}{1-q} \right)^T \leq \exp \left(  - \frac{q T}{1-q} \right)$.
Hence we can write
\begin{eqnarray}
K (1-  q)^T - \frac{K^2}{2} (1-2q)^T & = & \left( K(1-q)^T \right)
\left( 1 - \frac{K}{2} \left( \frac{1-2q}{1-q} \right)^T \right) \nonumber \\
& \geq & \exp \left( \ln K - \frac{q T}{1-q} \right) \left[ 1 - \frac{1}{2} \exp \left( \ln K - \frac{q T}{1-q} \right)  \right]. \label{eq:twoterms}
\end{eqnarray}
Now, this is precisely the quantity we need to control in the upper bound of  (\ref{eq:bonferroni}), 
taking $q = 1/(e(K-1))$.
Specifically, if we take $T = \lceil (1/q - 1) \ln K \rceil$, then $T \geq (1/q - 1) \ln K$, so that
then $1 \geq \exp(\ln K - q T/(1-q)) $, so the term in square brackets in (\ref{eq:twoterms}) is at least $1/2$.

Similarly since $T \leq (1/q - 1) \ln K + 1$, the $ \exp \left( \ln K - \frac{q T}{1-q} \right) \geq \exp( - q/(1-q))$,
which converges to 1 as $K \rightarrow \infty$ and hence $q \rightarrow 0$, and is certainly $\geq 2/3$ for $K \geq 3$.
\end{enumerate}
\end{proof}

These results can be 
 compared with the bound of  Chan \etal\!\!, Corollary \ref{cor:chanbd}. Again, in the sparsity regime
$K = N^{1-\beta}$, Lemma \ref{lem:asymp} shows that the error probability bound
behaves like $K^{-\delta'}$, so taking $\delta' = \delta/(1-\beta)$, the error probability bound behaves like
$N^{-\delta}$.

Corollary \ref{cor:chanbd} shows that to guarantee an error probability bound of $N^{-\delta}$ takes at most
$T = e \left( \delta + 1 \right) K \ln N$ tests, whereas Lemma \ref{lem:asymp} shows
at least $T = e \left( \delta + 1 - \beta \right) K \ln N$ tests are required.
In other words, for a given error probability, the upper and lower bound are separated by a constant additive gap of 
size $\beta/(1-\beta) e K \ln K$, again showing that (for fixed $K$) sparse problems are easier to solve.

Using Lemma \ref{lem:asymp}, we can show that in certain sparsity regimes, using
Bernoulli sampling suggests a
strict gap between the capacity of adaptive and non-adaptive group testing, assuming that the \algsss algorithm
is optimal.
\begin{theorem} \label{thm:asympsss}
Using any Bernoulli test design, taking 
\begin{equation} \label{eq:assssbd2}
\olc{\algsss}(\beta) = \frac{\beta }{(1-\beta) e \ln 2},
\end{equation}
  using the \algsss algorithm with 
$$ \frac{ \log_2 \binom{N}{K(N)}}{ T(N)} \geq \olc{\algsss}(\beta) + \epsilon,$$
has success probability less than $2/3$.
\end{theorem}
\begin{proof}
 Use the fact that by Lemma \ref{lem:binom} below
$$ \lim_{N \rightarrow \infty} \frac{\log_2 \binom{N}{K}}{K \ln K} = 
\lim_{N \rightarrow \infty} \frac{\log_2 \binom{N}{K}}{(1-\beta) K \ln N} = \frac{\beta}{(1-\beta) \ln 2}.$$
This choice of $\olc{\algsss}(\beta)$ ensures that if
$$ \olc{\algsss}(\beta) + \epsilon  \leq  \frac{ \log_2 \binom{N}{K(N)}}{T(N) } $$
then for $N$ sufficiently large,
$$ \frac{T}{e K \ln K}  \leq   \frac{e \olc{\algsss}(\beta) + \epsilon/2}{ e( \olc{\algsss} (\beta) + \epsilon)} < 1.$$
Combining Theorem \ref{prop:sample} and Lemma \ref{lem:asymp}, we can deduce that
the success probability $\pr(\suc) \leq 2/3$, so does not tend to $1$.
\end{proof}

When combined with the universal bound $R^*_{\mathtt{SSS}} \leq 1$, this proves Theorem \ref{sssthm}, that
  \[ R^*_{\mathtt{SSS}} \leq \frac{1}{e \ln 2} \min \left\{ 1, \frac{\beta}{1-\beta} \right\} . \]

Note that $\olc{\algsss}(\beta) < 1$ if and only if $\beta < \beta^* =  (e \ln 2)/(1 + e \ln 2) \simeq 0.653$.
This shows that the presence of an `adaptivity gap' may depend on the level of sparsity. 
That is, for $\beta > \beta^*$ (for sufficiently sparse problems) the information lower bound Theorem \ref{thm:basic}
dominates, and so we have no reason to think that the non-adaptive
capacity will be below $1$.
For $\beta < \beta^*$ (for less sparse problems), the bound from Theorem \ref{prop:sample} dominates,
and the capacity should be strictly less than 1. 



\section{Proofs: background probability facts} \label{sec:facts}
In order to analyse the probability that the \algdd algorithm succeeds, we
need to recall some facts from probability theory, including some  properties of the multinomial distribution.
\begin{lemma} \label{lem:multifacts}
For some $M \in \ZZ_+$, and some vector $\vc{p} = (p_1, \ldots, 
p_\ell)$ with $p_i \geq 0$ and 
$\sum_{i=1}^\ell p_i = 1$, suppose the vector $\vc{X}$ has
 multinomial probability
\begin{equation} \label{eq:multinomial}
\pr_{M; \vc{p}}(\vc{X} = \vc{x}) := \pr_{M; \vc{p}}(x_1, \ldots, x_\ell)
= \frac{ M!}{ \prod_{i=1}^{\ell} x_i!} \prod_{i=1}^{\ell} p_i^{x_i}
\qquad \text{for $\sum_{i=1}^\ell x_i = M$}.
\end{equation}
Then
\begin{enumerate}
\item \label{it:zeroprob} For any collection $\CC$ of indices, the 
$\pr_{M;\vc{p}} \left( \bigcap_{i \in \CC} \{ X_i = 0 \} \right) 
= \left( 1- \sum_{i \in \CC} p_i \right)^M$. 
\item \label{it:marginal}
For any $s$, the marginal distributions are binomial, in that
$$\pr_{M; \vc{p}}(X_s = x_s) = \binom{M}{x_s} p_s^{x_s} 
(1-p_s)^{M-x_s}.$$
\item \label{it:cond} For any $s$, write $\vc{u}^{(s)} = 
(u_1, \ldots, u_{s-1}, u_{s+1}, \ldots,
u_\ell)$ for the vector $\vc{u}$ with the $s$th component removed. 
Then the conditional distribution given $X_s$ is still multinomial, in that
$$ \pr_{M; \vc{p}}( \vec X^{(s)} = \vec x^{(s)} | X_s = x_s) = 
 \pr_{M-x_s; \vc{p^{(s)}}}( \vec x^{(s)}),$$
where $p^{(s)}_i = p_i/(1-p_s)$ for $i \neq s$.
\item \label{it:split}
Given $\vc{X} \sim \pr_{M; \vc{p}}$, split class $i$ into new classes 
$i+$ and $i-$, such that (independently)  each member of class $i$
enters class $i+$ with probability $Q$ and otherwise enters
class $i-$. Then
$$ \vc{X}' = ( X_1, \ldots, X_{i-1}, X_{i+}, X_{i-}, X_{i+1}, \ldots, X_n)
\sim \pr_{M; \vc{p}'},$$
where $\vc{p}' = ( p_1, \ldots, p_{i-1}, p_i Q, p_i (1-Q), p_{i+1}, \ldots, p_n)$.
\end{enumerate}
\end{lemma}

\begin{proof}
The first two facts follow from the multinomial theorem, which says that
for any $m$:
$$ \sum_{k_1 + \ldots + k_m = L}
\frac{ L!}{ \prod_{i=1}^{m} k_i!} \prod_{i=1}^m p_i^{k_i} =
( p_1 + \ldots + p_m)^L,$$
and the third follows by rearranging. The last fact follows since we can take the 
ratio of $\pr_{M; \vc{p}}$ and $\pr_{M; \vc{p}'}$ to obtain
 $$ \binom{x_i}{x_{i+}} Q^{x_{i+}} (1-Q)^{x_{i-}}, $$
as required. 
\end{proof}

Using this, we can derive an expression for the success probability in a particular `symmetric' case, in terms of
the $\phi_K$ function of  (\ref{eq:phi}):
$$ \phi_K(q,T) := \sum_{\ell = 0}^K (-1)^\ell \binom{K}{\ell} (1 - \ell q)^T\ .$$

\begin{lemma} \label{lem:maskprob}
Fix $K \geq 1$, $M \geq 1$ and $0 \leq q \leq 1/K$. and let
$(X_1, \ldots, X_k, X')$ have multinomial probability $\pr_{M, \vc{q}}$,
where the first $K$ components of $\vc{q}$ are identical, with $\vc{q} = (q, q, \ldots, q, 1 - K q)$.
 Then 
\begin{equation} \label{eq:phiident} \pr\left( \bigcap_{i=1}^K \{ X_i \neq 0 \} \right) = \phi_K(q,M), \end{equation}
and $\phi_K$ satisfies the bounds 
\begin{equation} \label{eq:bonferroni}
 \max\left\{0, 1 - K (1-  q)^M \right\} \leq \phi_K(q,M) \leq 1 - K (1-  q)^M + \frac{K^2}{2} (1-2q)^M.
\end{equation}
\end{lemma}

\begin{proof} First, notice that for any set $\CC$ with $|\CC| = \ell$, Lemma \ref{lem:multifacts}.\ref{it:zeroprob} gives
\begin{equation} \label{eq:probinter} \pr_{M;\vc{q}} \left( \bigcap_{i \in \CC} \{ X_i = 0 \} \right) 
= \left( 1-  \ell q \right)^M.\end{equation}
Then we prove the identity  (\ref{eq:phiident}) since
\begin{align}
\pr\left( \bigcap_{i=1}^K \{ X_i \neq 0 \} \right)  &= 1 - \pr\left( \bigcup_{i=1}^K\{ X_i=0\} \right) \notag \\
  &= 1 - \sum_{\ell=1}^K  \sum_{\CC: |\CC| = \ell} \pr\left( \bigcap_{i \in \CC}^K\{ X_i=0\} \right) \label{eq:ident2} \\
  &= 1 - \sum_{\ell=1}^K  \binom{K}{\ell}  \left( 1-  \ell q \right)^M \label{eq:ident3}.
\end{align}
Here (\ref{eq:ident2}) is simply an application of  the inclusion-exclusion formula (see for example \cite[Chapter IV, Equation (1.5)]{feller}),
and (\ref{eq:ident3}) follows using (\ref{eq:probinter}).

Clearly, since $\phi_K$ is a probability, we must have  $ \phi_K(q,M) \geq 0$.
The remaining bounds on $\phi_K$ follow from applications of the Bonferroni inequalities (see for example \cite[Chapter IV, Equation (5.6)]{feller}) These results  state that
(a) we can lower bound the expression (\ref{eq:ident2}) by truncating the sum at $\ell = 1$, and (b) we can upper bound the expression by truncating the sum at $\ell =2$.
In each case (\ref{eq:bonferroni}) follows by again using (\ref{eq:probinter}). 
 \end{proof}

\begin{lemma}
\label{lem:nonincr}
For fixed $T$, the function
$$ \phi_K(q,T) := \sum_{\ell = 0}^K (-1)^\ell \binom{K}{\ell} (1 - \ell q)^T\ ,$$ 
 is increasing in $q$.
\end{lemma}

\begin{proof}
The key is to observe that direct calculation gives that
\begin{align}
	\df{\partial}{\partial q}
\phi_K(q,T) &= \sum_{\ell =0}^K (-1)^\ell \binom{K}{\ell} \left(
- \ell  T (1- \ell q)^{T-1} \right) \notag \\
	& = T K  \sum_{\ell=1}^K
(-1)^{\ell-1} \binom{K-1}{\ell-1} (1- \ell q)^{T-1}  \notag \\
& = T K (1-q)^{T-1}  \sum_{\ell=1}^K
(-1)^{\ell-1} \binom{K-1}{\ell-1} \left( 1- (\ell-1 ) \frac{q}{1-q}\right)^{T-1}  \notag \\
& =   T K (1-q)^{T-1}  \phi_{K-1} \left( \frac{q}{1-q}, T-1 \right) \notag \\
& \geq 0, \label{eq:useful2}
\end{align}
where we have used the facts that  $$\ell \binom{K}{\ell}=K \binom{K-1}{\ell-1},$$ 
and
\[ (1 - \ell q) = (1-q) \left( 1 - (\ell - 1) \frac{q}{1-q} \right).\]
 The positivity of (\ref{eq:useful2}) follows since $\phi_K(q,T)$ is a probability,
 and hence positive for any choice of $K$, $q$ and $T$ (see (\ref{eq:bonferroni}) above).
\end{proof}

\begin{theorem}[Chernoff-Hoeffding theorem \cite{hoeffding1963}]
\label{thm:chernoff}
Let $X_1, X_2, \ldots$ be independent 
and identically distributed random variables with $\E X_1=p$. Then, for all $0<\veps<1-p$,
\begin{align}
\label{eq:chernoff1}
\pr\left(\df{1}{m}\sum_{i=1}^m X_i > p+\veps\right)  	& \leq e^{-mD(p+\veps\|p)}
\end{align}
\end{theorem}

\section*{Acknowledgments}
We would like to thank the anonymous referees for their careful reading of this paper, and for their suggestions of how to present  our work.


\end{document}